\makeatletter
\def\input@path{{styles/}{../styles/}}
\makeatother

\ifx\CONF\undefined%
   \documentclass[11pt]{article}%
   \def\UseBibLatex{1}%
   \providecommand{\ConfVer}[1]{}%
   \providecommand{\NotConfVer}[1]{#1}%
\else%
\documentclass[a4paper,UKenglish,numberwithinsect,cleveref,
autoref]{lipics-v2021}
\def\ProofInAppendix{1}

\providecommand{\ConfVer}[1]{#1}%
\providecommand{\NotConfVer}[1]{}%
\fi

\usepackage{todonotes}

\NotConfVer{%
   \def\UseBibLatex{1}%
}

\NotConfVer{%
   \usepackage{nicefrac}
   \usepackage{amssymb}%
   \usepackage[in]{fullpage}%
}%
\usepackage{xcolor}%
\usepackage{fancyvrb}
\usepackage{proofapnd}
\usepackage{xspace,mleftright}

\usepackage{hyperref}%

\definecolor{LinkColor}{rgb}{0.5,0.0,0.0}
\hypersetup{%
      unicode,
      breaklinks,%
      colorlinks=true,%
      urlcolor=[rgb]{0.25,0.0,0.0},%
      linkcolor=LinkColor,%
      citecolor=[rgb]{0,0.2,0.445},%
      filecolor=[rgb]{0,0,0.4},
      anchorcolor=[rgb]={0.0,0.1,0.2}%
}
\usepackage[ocgcolorlinks]{ocgx2}

\newcommand{\HLinkY}[2]{\hyperref[#2]{#1}}
\newcommand{\HLink}[2]{\hyperref[#2]{#1 \ref{#2}}}
\newcommand{\HLinkSuffix}[3]{\hyperref[#2]{#1\ref*{#2}{#3}}}

\newcommand{\seclab}[1]{\label{sec:#1}}
\newcommand{\secref}[1]{\HLink{Section}{sec:#1}}

\newcommand{\apndlab}[1]{\label{apnd:#1}}
\newcommand{\apndref}[1]{\HLink{Appendix}{apnd:#1}}

\newcommand{\obslab}[1]{\label{observation:#1}}
\newcommand{\obsref}[1]{\HLink{Observation}{observation:#1}}

\newcommand{\remlab}[1]{\label{rem:#1}}
\newcommand{\remref}[1]{\HLink{Remark}{rem:#1}}%

\newcommand{\lemlab}[1]{\label{lemma:#1}}
\newcommand{\lemref}[1]{{\expandafter\HLink{Lemma}{lemma:#1}}}%
\newcommand{\Xlemref}[1]{{\noexpand\HLink{Lemma}{lemma:#1}}}%

\newcommand{\figlab}[1]{\label{fig:#1}}
\newcommand{\figref}[1]{\HLink{Figure}{fig:#1}}

\newcommand{\thmlab}[1]{{\label{theo:#1}}}
\newcommand{\thmref}[1]{\HLink{Theorem}{theo:#1}}
\newcommand{\thmrefY}[2]{\HLinkY{#2}{theo:#1}}

\newcommand{\corlab}[1]{\label{cor:#1}}

\newcommand{\defnlab}[1]{\label{defn:#1}}
\newcommand{\defnref}[1]{\HLink{Definition}{defn:#1}}%

\providecommand{\deflab}[1]{\label{def:#1}}
\newcommand{\defref}[1]{\HLink{Definition}{def:#1}}
\newcommand{\defrefY}[2]{\hyperref[def:#2]{#1}}

\providecommand{\eqlab}[1]{}%
\renewcommand{\eqlab}[1]{\label{equation:#1}}
\newcommand{\Eqref}[1]{\HLinkSuffix{Eq.~(}{equation:#1}{)}}

\providecommand{\ineqlab}[1]{}%
\renewcommand{\ineqlab}[1]{\label{equation:#1}}

\newcommand{\hyperrefB}[2]{%
   \begingroup%
   \hypersetup{linkcolor=black}%
   \hyperref[#1]{#2}%
   \endgroup%
   \hypersetup{linkcolor=LinkColor}%
}

\newcommand{\prophetOnly}{\ensuremath{\textsc{Proph}}}

\newcommand{\prophet}{\ensuremath{%
      \hyperrefB{def:prop:h:o:m}{\prophetOnly}}}

\newcommand{\oracleOnly}{\color{black}{\Oracle}}
\newcommand{\oracle}{\hyperrefB{def:prop:h:o:m}{\oracleOnly}\xspace}

\newcommand{\pbmOnly}{\ensuremath{\mathbb{P}_{\max}}}

\newcommand{\pbm}{{{\hyperrefB{def:PbM}{\ensuremath{\mathbb{P}_{\max}}}}}%
   \xspace}
\newcommand{\roeOnly}{{\textsf{RoE}}\xspace}
\newcommand{\roe}{\hyperrefB{def:RoE}{\textsf{RoE}}\xspace}

\renewcommand{\th}{th\xspace}

\newcommand{\cdf}{\textsf{cdf}\xspace}

\newcommand{\iidOnly}{\textsf{IID}\xspace}
\newcommand{\noniidOnly}{\textsf{non-IID}\xspace}

\newcommand{\iid}{\hyperrefB{def:toiidornot}{\color{black}{%
         \textsf{IID}}}\xspace}
\newcommand{\noniid}{\hyperrefB{def:toiidornot}{\color{black}{%
         \ensuremath{{\textsf{non-IID}}}}}\xspace}

\providecommand{\BibLatexMode}[1]{}
\providecommand{\BibTexMode}[1]{#1}

\ifx\UseBibLatex\undefined%
  \renewcommand{\BibLatexMode}[1]{}
  \renewcommand{\BibTexMode}[1]{#1}
\else
  \renewcommand{\BibLatexMode}[1]{#1}
  \renewcommand{\BibTexMode}[1]{}
\fi

\BibLatexMode{%
   \usepackage[bibencoding=utf8,style=alphabetic,backend=biber]{biblatex}%
   \usepackage{sariel_biblatex}%
}

\NotConfVer{%
   \usepackage{amsmath}

\usepackage[amsmath,thmmarks]{ntheorem}%
\theoremseparator{.}%

\theoremstyle{plain}%
\newtheorem{theorem}{Theorem}[section]

\newtheorem{lemma}[theorem]{Lemma}

\newtheorem{observation}[theorem]{Observation}

\theoremstyle{plain}%
\theoremheaderfont{\sf}%
\theorembodyfont{\upshape}%
\newtheorem*{remark:unnumbered}[FakeCounter]{Remark}%
\newtheorem{remark}[theorem]{Remark}%
\newtheorem{definition}[theorem]{Definition}
\newtheorem*{def:unnumbered}[FakeCounter]{Definition}
\newtheorem{defn}[theorem]{Definition}

\newcommand{\myqedsymbol}{\rule{2mm}{2mm}}

\definecolor{blue25emph}{rgb}{0, 0, 11}
\newcommand{\emphic}[2]{%
   \textcolor{blue25emph}{%
      \textbf{\emph{#1}}}%
   \index{#2}}

\newcommand{\emphi}[1]{\emphic{#1}{#1}}

\definecolor{almostblack}{rgb}{0, 0, 0.3}
\providecommand{\emphw}[1]{{\textcolor{almostblack}{\emph{#1}}}}%

\theoremheaderfont{\em}%
\theorembodyfont{\upshape}%
\theoremstyle{nonumberplain}%
\theoremseparator{}%
\theoremsymbol{\myqedsymbol}%
\newtheorem{proof}{Proof:}%

}

\ConfVer{%
   \usepackage{vasilis-lipics}
}%

\providecommand{\etal}{\textit{et~al.}\xspace}
\renewcommand{\etal}{\textit{et~al.}\xspace}

\newcommand{\ZZ}{{Z}}

\newcommand{\ProbLTR}{\mathbb{P}}%
\renewcommand{\Pr}{\ProbLTR}

\newcommand{\Prob}[1]{\mathop{\ProbLTR} \mleft[ #1 \mright]}%
\newcommand{\ProbCond}[2]{\mathop{\ProbLTR}\!\left[%
       #1 \;\middle\vert\; #2 \right]}
\newcommand{\prnCond}[2]{\mleft( #1 \;\middle\vert\; #2 \mright)}

\newcommand{\ExChar}{\mathbb{E}}%
\newcommand{\Ex}[2][\!]{\mathop{\ExChar}#1\pbrcx{#2}}

\newcommand{\Var}[1]{\mathop{\mathbb{V}}\!\pbrcx{#1}}

\newcommand{\pbrcx}[1]{\left[ {#1} \right]}%

\newcommand{\f}[2]{\ensuremath{\nicefrac{#1}{#2}}\xspace}
\newcommand{\prn}[1]{\mleft(#1\mright)}%
\newcommand{\pth}[1]{\prn{#1}}%

\newcommand{\BigOmega}{\operatorname{\Omega}}
\newcommand{\BigO}{\mathcal{O}}
\newcommand{\SmallO}{o}
\newcommand{\Oracle}{\mathcal{O}}
\newcommand{\alg}{{ALG}\xspace}

\newcommand{\Set}[2]{\left\{ #1 \;\middle\vert\; #2 \right\}}
\newcommand{\set}[1]{\left\{ {#1} \right\}}

\newcommand{\cM}{\mathcal{M}}
\newcommand{\cA}{\mathcal{A}}
\newcommand{\cB}{\mathcal{B}}
\newcommand{\cD}{\mathcal{D}}
\newcommand{\cX}{\mathcal{X}}

\newcommand{\eps}{{\varepsilon}}%

\newcommand{\NN}{\mathbb{N}}

\newcommand*\diff{\mathop{}\!\mathrm{d}}

\newcommand{\dif}[1]{\mathop{d #1}}

\newcommand{\abs}[1]{\left| {#1} \right|}%

\usepackage{stmaryrd}%
\providecommand{\IntRange}[1]{\mleft\llbracket #1 \mright\rrbracket}
\newcommand{\IRX}[1]{\IntRange{#1}}%
\newcommand{\IRY}[2]{\left\llbracket #1:#2 \right\rrbracket}

\newcommand{\Pois}{\operatorname{Pois}}

\newcommand{\QY}[2]{\mathsf{q}_{#1}\prn{#2}}
\newcommand{\dQY}[2]{\mathsf{q}_{#1}'\prn{#2}}

\newcommand{\Shard}{\mathbf{H}}%
\newcommand{\XX}{\mathbf{X}}%
\newcommand{\Seq}{\mathbf{S}}%

\newcommand{\cardin}[1]{\left| {#1} \right|}%

\newcommand{\atgen}{\symbol{'100}}%

\newcommand{\SarielThanks}[1]{%
   \thanks{%
      Department of Computer Science; %
      University of Illinois; %
      201 N. Goodwin Avenue; %
      Urbana, IL, 61801, USA; %
      \href{mailto:spam@illinois.edu}{sariel@illinois.edu}; %
      \url{http://sarielhp.org/}. %
   #1%
   }%
}

\usepackage[inline]{enumitem}

\newlist{compactenumA}{enumerate}{5}%
\setlist[compactenumA]{topsep=0pt,itemsep=-1ex,partopsep=1ex,parsep=1ex,%
   label=(\Alph*)}%

\newlist{compactenuma}{enumerate}{5}%
\setlist[compactenuma]{topsep=0pt,itemsep=-1ex,partopsep=1ex,parsep=1ex,%
   label=(\alph*)}%

\newlist{compactenumI}{enumerate}{5}%
\setlist[compactenumI]{topsep=0pt,itemsep=-1ex,partopsep=1ex,parsep=1ex,%
   label=(\Roman*)}%

\newlist{compactenumi}{enumerate}{5}%
\setlist[compactenumi]{topsep=0pt,itemsep=-1ex,partopsep=1ex,parsep=1ex,%
   label=(\roman*)}%

\newlist{compactitem}{itemize}{5}%
\setlist[compactitem]{topsep=0pt,itemsep=-1ex,partopsep=1ex,parsep=1ex,%
   label=\ensuremath{\bullet}}%

\newcommand{\FaroukThanks}[1]{%
   \thanks{Department of Computer Science; University of Illinois; 201
      N. Goodwin Avenue; %
      Urbana, IL, 61801, %
      USA; %
      {\tt eyharb2\atgen{}illinois.edu}; %
      {\tt \url{https://farouky.github.io/}.}%
      #1}%
}

\newcommand{\VasilisThanks}[1]{%
   \thanks{Department of Computer Science; University of Chile; Chile; %
      {\tt livanos3@illinois.edu}; %
      {\tt \url{https://livanos3.web.engr.illinois.edu/}.}%
      #1}%
}

\newcommand{\SaveContent}[2]{%
   \expandafter\newcommand{#1}{#2}%
}

\numberwithin{figure}{section}%
\numberwithin{table}{section}%
\numberwithin{equation}{section}%

\usepackage{titlesec}%
\titlelabel{\thetitle. }%

\newcommand{\si}[1]{#1}
\providecommand{\TPDF}[2]{\texorpdfstring{#1}{#2}}

\newcommand{\I}{\mathcal{I}}
\newcommand{\J}{\mathcal{J}}

\newcommand{\Term}[1]{\textsf{#1}}
\newcommand{\TermI}[1]{\Term{#1}\index{#1@\Term{#1}}}

\newcommand{\YES}{\TermI{YES}\xspace}
\newcommand{\NO}{\TermI{NO}\xspace}

\newcommand{\LHoptial}{L'H\^opital\xspace}

\BibLatexMode{%
   \bibliography{prophet_oracle} }

\begin{document}

\title{Oracle-Augmented Prophet Inequalities}

\ConfVer{%
   \author{Anonymous Author(s)}{Anonymous Affiliation(s)}{}{}{} }
\NotConfVer{%

   \author{%
      Sariel Har-Peled%
      \SarielThanks{Work on this paper was partially supported by NSF
         AF award CCF-2317241.}%
      \and%
      Elfarouk Harb%
      \FaroukThanks{}%
      \and%
      Vasilis Livanos%
      \VasilisThanks{}%
   }%
}%

\ConfVer{%
   \authorrunning{Anonymous Author(s)}%
   \Copyright{CC-BY} }

\ConfVer{%
} \ConfVer{%
   \ccsdesc[500]{Theory of computation~Algorithmic mechanism design}

   \keywords{prophet inequalities, predictions, top-$1$-of-$k$ model}

   \category{}

   \relatedversion{} }

\ConfVer{%
   \EventEditors{John Q. Open and Joan R. Access}%
   \EventNoEds{2}%
   \EventLongTitle{International Colloquium on Automata, Languages,
      and Programming 2024}%
   \EventShortTitle{ICALP 2024} \EventAcronym{ICALP}%
   \EventYear{2024}%
   \EventDate{December 24--27, 2024}%
   \EventLocation{Little Whinging, United Kingdom}%
   \EventLogo{}%
   \SeriesVolume{42}%
   \ArticleNo{23}%
}%

\maketitle

\begin{abstract}
    In the classical prophet inequality settings, a gambler is given a
    sequence of $n$ random variables $X_1, \dots, X_n$, taken from
    known distributions, observes their values in this (potentially
    adversarial) order, and select one of them, immediately after it
    is being observed, so that its value is as high as possible. The
    classical \emph{prophet inequality} shows a strategy that
    guarantees a value at least half of that an omniscience prophet
    that picks the maximum, and this ratio is optimal.

    Here, we generalize the prophet inequality, allowing the gambler
    some additional information about the future that is otherwise
    privy only to the prophet. Specifically, at any point in the
    process, the gambler is allowed to query an oracle
    $\mathcal{O}$. The oracle responds with a single bit answer: YES
    if the current realization is greater than the remaining
    realizations, and NO otherwise. We show that the oracle model with
    $m$ oracle calls is equivalent to the \textsc{Top-$1$-of-$(m+1)$}
    model when the objective is maximizing the probability of
    selecting the maximum. This equivalence fails to hold when the
    objective is maximizing the competitive ratio, but we still show
    that any algorithm for the oracle model implies an equivalent
    competitive ratio for the \textsc{Top-$1$-of-$(m+1)$} model.

    We resolve the oracle model for any $m$, giving tight lower and
    upper bound on the best possible competitive ratio compared to an
    almighty adversary. As a consequence, we provide new results as
    well as improvements on known results for the
    \textsc{Top-$1$-of-$m$} model.
\end{abstract}

\section{Introduction}
\seclab{introduction}

The field of optimal stopping theory concerns optimization settings
where one makes decisions in a sequential manner, given imperfect
information about the future, in a bid to maximize a reward or
minimize a cost. A canonical setting in this area is the \emph{prophet
   inequality} \cite{ks-sfv-77, ks-sapfv-78}. In these settings, a
gambler is presented with rewards $X_1, \dots, X_n$, one after the
other, drawn independently from known distributions. Upon seeing a
reward $X_i$, the gambler must immediately make an irrevocable
decision to either accept $X_i$, in which case the process ends, or to
reject $X_i$ and continue, losing the option to select $X_i$ in the
future. The goal of the gambler is to maximize the selected reward
comparing against a \textit{prophet} who knows all realizations in
advance and selects the maximum realized reward. Throughout, we
assume, without loss of generality, that $X_1, \dots, X_n$ are
continuous random variables.

The performance of the gambler can be measured in terms of several
objectives. A common metric used in the literature is the
\emph{competitive ratio}, which is also known as the {\emph{Ratio of
      Expectations (\roe)}} (see~\defref{RoE}). Another common
distinction is between the case in which the given distributions are
different and the case in which they are identical. For the former,
Krengel \etal \cite{ks-sfv-77, ks-sapfv-78} showed an optimal strategy
that is \f{1}{2}-competitive. In this setting, the optimal competitive
ratio can be achieved by simple, single-threshold algorithms
\cite{s-ctsrm-84, kw-mpiam-19}. For \iid and \noniid random variables,
Hill and Kertz \cite{hk-csrse-82} initially gave a
$(1 - \f{1}{e})$-competitive algorithm. This was improved to
$\approx 0.738$ \cite{aeehk.ea-b1op-17} and later $\approx 0.745$
\cite{cfhov-ppmot-21}, which is tight, due to a matching upper bound
\cite{hk-csrse-82, k-srsei-86}.

Another relevant metric, introduced by Gilbert and Mosteller
\cite{gm-rms-66} for \iid random variables, is that of maximizing the
\emph{Probability of selecting the Maximum realization} (\pbm) -
see~\defref{PbM}. For this objective and \iid random variables,
Gilbert and Mosteller \cite{gm-rms-66} gave an algorithm that achieves
a probability of $\approx 0.58$, which is the best possible. Later,
Esfandiari, Hajiaghayi, Lucier and Mitzenmacher \cite{ehlm-ps-17}
studied the same objective for general random variables, obtaining a
tight probability equal to $\f{1}{e}$ when the random variables arrive
in adversarial order and $0.517$ when the random variables arrive in
random order. The latter case was recently improved to the tight
$\approx 0.58$ by Nuti \cite{n-spd-22}, showing that the \iid setting
is not easier than the \noniid setting with random order. In this
paper, we introduce a new model as a means to study variations of both
the \iid and the general settings, for both the \roe and \pbm
objectives.

A setting that is related to ours is the \textsc{Top-$1$-of-$m$}
model, formally introduced by Assaf and Samuel-Cahn \cite{as-srpim-00}
for \iid random variables, although it had been studied initially by
Gilbert and Mosteller \cite{gm-rms-66}. In this setting, the algorithm
is allowed to select $m \geq 1$ values, but the value it gets judged
by is the maximum selected value. Gilbert and Mosteller
\cite{gm-rms-66} gave numerical approximations of the~\pbm objective
for $2 \leq m \leq 10$, using a simple, single-threshold
algorithm. Later, Assaf and Samuel-Cahn \cite{as-srpim-00} studied the
\roe objective for general distributions and gave an elegant and
simple $\prn{1 - \f{1}{m+1}}$-competitive algorithm. This was improved
\cite{ags-rpiwm-02} by bounding the competitive ratio of the optimal
algorithm by a recursive differential equation. They gave numerical
approximations for $2 \leq m \leq 5$, but studying the asymptotic
nature of the constants for large $m$ turned out to be difficult. Ezra
\etal \cite{efn-pso-18} later revisited the problem and gave a new
algorithm for large $m$ that is
$1 - \BigO\prn{e^{-\f{m}{6}}}$-competitive for the same problem. This
improves the error term from \cite{ags-rpiwm-02} from linear in $m$ to
exponential in $m$. Harb \cite{h-fbcps-23} recently improved this into
a $1 - e^{-m W_0\prn{\frac{\sqrt[m]{m!}}{m}}}$-competitive algorithm,
where $W_0$ is the Lambert-$W$ function\footnote{The Lambert-$W$
   function is $W_0(x)$ defined as the solution $y$ to the equation
   $y e^y = x$.}, and improved the lower bound for $m = 2$
separately. However, the asymptotic nature of this function is
difficult to analyze.

\paragraph{Type of Adversary.}

In the context of prophet inequalities, the distinction between an
offline adversary and an almighty adversary is crucial to
understanding the competitive ratio bounds of prophet inequalities. An
offline adversary, often considered less powerful, observes the
\emph{distributions} of the random variables, and chooses an
adversarial order based on the distributions. An almighty adversary is
stronger:

\begin{defn}
    \deflab{mighty}%
    An \emphi{almighty adversary} possesses complete information,
    including the algorithm's random decisions, and can thus tailor
    the ordering of the sequence to worst-case scenarios with perfect
    foresight. In particular, the almighty adversary observes the
    gambler's algorithm, and the values of $X_1, \dots, X_n$, then
    chooses a permutation order $X_{\sigma(1)}, \dots, X_{\sigma(n)}$
    to show the algorithm the values in that order.
\end{defn}

Consequently, while both types of adversaries present different levels
of challenge, the almighty adversary sets a far stricter benchmark,
typically leading to lower competitive ratios for prophet
inequalities. Unless stated otherwise, we work with the almighty
adversary. For a discussion on the (very subtle) differences between
the two adversaries for our model, see \apndref{app:final:discussion}.

\paragraph*{Model.}
We introduce a new model that generalizes the standard prophet
inequality setting, and analyze it as a means to obtain new results
and improvements in the \textsc{Top-$1$-of-$m$} model. Our model
allows the algorithm some information about the future that is
otherwise privy only to the prophet. Specifically, at any point in the
process, upon seeing a reward $X_i$, the algorithm is allowed to query
an oracle $\Oracle$. The oracle $\Oracle$ responds with a single bit
answer: \textsc{YES} if the current realization is larger than the
remaining realizations, i.e., $X_i > \max_{j = i + 1}^n X_j$ and
\textsc{NO} otherwise. In other words, the oracle $\Oracle$ informs
the algorithm it should select $X_i$, or reject it, because there is a
reward coming up that is at least as good\footnote{There are
   \emph{very subtle} differences between an oracle that answers $>$
   queries, vs $\geq$ queries. See \apndref{app:final:discussion} for a
   discussion on this. In particular, the $>$ oracle is weaker than
   the $\geq$ oracle; imagine a stream of $1$ values, the $>$ oracle
   will always answer \NO, while the $\geq$ oracle answers \YES on the
   first query}. Clearly, with no queries available, one recovers the
classical prophet inequality setting, whereas with $n - 1$ queries,
the strategy of using a query on every $X_i$, for
$i = 1, \dots, n - 1$, leads to the algorithm selecting the highest
realization always. Thus, this model interpolates nicely between the
two extremes of full or no information about the future.

In this paper, we consider the following different settings.

\begin{definition}
    \deflab{RoE}%
    The competitive ratio or \emph{Ratio of Expectations} is denoted
    by $\roe$. An algorithm \alg is \emph{$\alpha$-competitive}, for
    $\alpha \in [0,1]$, if
    $\Ex{\alg} \geq \alpha \cdot \Ex{\max_i X_i}$, and $\alpha$ is
    called the \emph{competitive ratio}.
\end{definition}

\begin{definition}
    \deflab{PbM}
    The \emph{Probability of selecting the Maximum} realization is
    denoted by $\pbm$. An algorithm \alg achieves a~$\pbm$ of $\alpha$
    if it returns a value $v$ such that $\Prob{v = \ZZ} \geq \alpha$,
    where $\ZZ = \max\set{X_1, \hdots, X_n}$. Note that in some works
    (for example \cite{gm-rms-66}), the notation $PbM$ has also been
    used.
\end{definition}

\begin{definition}
    \deflab{toiidornot}
    We use the term \iid to refer to the setting where
    $X_1, \dots, X_n$ are independent and identically distributed
    random variables. We use \noniid to refer to the more general
    setting where $X_1, \dots, X_n$ are independent, but not
    necessarily identical.
\end{definition}

\begin{definition}
    \deflab{prop:h:o:m}%
    We use $\prophet_m$ to refer to the \textsc{Top-$1$-of-$m$} model,
    in which the algorithm is allowed to choose up to $m$ values, and
    its payoff is the maximum of the chosen values. We use $\oracle_m$
    refers to our oracle model where the algorithm has access to $m$
    oracle calls, and can only select one value.
\end{definition}

Note that it makes sense to compare the model $\prophet_{m+1}$ to
$\oracle_m$ since in the former, the algorithm can choose $m+1$
values, where as the later can ask the oracle $m$ times, and then
choose an item. To help distinguish between the different settings, we
denote each model as $\cM(x, y, z)$, where
\begin{compactitem}
    \item $x$ is either $\prophet_m$ or $\oracle_m$ with $m \in \NN$,
    \item $y$ is either \iid or \noniid, and
    \item $z$ is either $\pbm$ or $\roe$.
\end{compactitem}

\paragraph*{Motivation.}

Our oracle choice was driven by our initial effort to reformulate the
Top-$1$-of-$m$ model in order to get a better understanding of that
settings and improve the known bounds. While we initially thought the
two models are the same, we later observed the subtle differences
between the two models. Thankfully, it turns out that one can still
use this oracle model to study the Top-$1$-of-$m$ model and improve
the previously known bounds, which was our original goal. We did
consider a few other oracle models, which we briefly mention here: if
the oracle predicts the maximum value, then there is a trivial
solution of asking the oracle at $X_1$, and just waiting until the
maximum value $v$ arrives. Another option is for the oracle to predict
a \emph{range} for the maximum value, but formalizing this in a more
general setting turns out to be difficult without assuming something
about the support of each random variable. While we explored various
other oracle formulations, we chose the simplest version that improved
the lower bound for the Top-$1$-of-$m$ model and for which we can also
achieve a tight competitive ratio. We leave exploring more complex
oracle models for future work.

\subsection{Our contributions}

In this paper, we study the oracle model for independent random
variables following identical or general distributions with the~\pbm
and~\roe objectives and make the following contributions:
\begin{compactenumI}%
    \smallskip%
    \item We establish an \emph{equivalence} between the oracle model
    and the \textsc{Top-$1$-of-$m$} model for the~\pbm objective.

    \smallskip%
    \item We show that this equivalence fails to hold for the~\roe
    objective and that the best-possible competitive ratios in the two
    settings are quite separated. However, we show that guarantees
    for~\roe in the oracle model translate to guarantees in the
    \textsc{Top-$1$-of-$m$} model, thus further motivating our study
    of the oracle model.

    \smallskip%
    \item We resolve the oracle model $\cM(\oracle_m, \noniid, \roe)$
    by presenting a single-threshold algorithm. Our algorithm achieves
    a competitive ratio of $1 - e^{- \xi_m}$ for general $m$, where
    $\xi_m$ is the unique \textit{positive} solution\footnote{In
       \secref{non-iid}, we prove that there is indeed a unique
       positive solution.} to the equation
    $1 - e^{-\xi_m} = \frac{\Gamma(m+1,
       \xi_m)}{m!}$\footnote{$\Gamma(n,x) = \int_x^\infty {t^{n-1}
          e^{-t} \dif t}$ denotes the upper incomplete gamma
       function}. Furthermore, we show that this lower bound is
    optimal by showing a construction that yields an equal upper
    bound. Since we showed that lower bound guarantees for
    $\cM(\oracle_m, \noniid, \roe)$ also hold for the alternative
    settings $\cM(\prophet_{m+1}, \noniid, \roe)$, this strictly
    improves the current state of the art bounds of \cite{h-fbcps-23},
    even though the guarantees are obtained in the weaker oracle
    model.

    \smallskip%
    \item We give a single-threshold algorithm for the oracle model
    and the~\pbm objective $\cM(\oracle_m, \iid, \pbm)$ that achieves
    a $1 - \BigO{(m^{-\f{m}{5}})}$ probability of selecting the
    maximum, as well as providing an upper bound that is
    asymptotically (almost) tight. To the best of our knowledge, this
    is the first result for the \pbm objective and general $m$ in the
    well studied \textsc{Top-$1$-of-$m$} model. Our algorithm achieves
    a probability of $\approx 0.797$ even with $m = 1$ calls to the
    oracle, a significant improvement on the $\approx 0.58$ achieved
    without oracle calls \cite{gm-rms-66}.
\end{compactenumI}
\smallskip%
As discussed earlier, the main motivation behind our oracle model
comes from our first two results which relate it to the
\textsc{Top-$1$-of-$m$} model.

\subsection{Our results in detail}

\subsubsection{Equivalence of models as far as \pbmOnly}

In \thmref{equivalence} we prove that $\cM(\oracle_m, y, \pbm)$ model
is equivalent to the $\cM(\prophet_{m+1}, y, \pbm)$ model, where
$y =\iid$ or \noniid.  In other words, the best algorithms in these
models achieve the same probability of realizing the maximum.

This result might not seem that surprising due to the apparent
similarity of the two models. However, thinking about the
\textsc{Top-$1$-of-$m$} setting from the viewpoint of oracle calls
allows for a different perspective that we exploit in our analysis.

\subsubsection{Difference of the models as far as \roeOnly}

Perhaps more surprisingly, our oracle model and the
\textsc{Top-$1$-of-$m$} model stop being equivalent when one considers
the~\roe objective. The oracle model is strictly weaker.

Specifically, in \thmref{second-reduction}, we prove that there exists
a prophet inequality instance, and an algorithm $\cA$ for
$\cM(\prophet_{m+1}, \noniid, \roe)$ instance for which no algorithm
for $\cM(\oracle_m, \noniid, \roe)$ can achieve the same competitive
ratio as that of $\cA$.  Furthermore, any algorithm for
$\cM(\oracle_m, y, \roe)$ can be modified to be an algorithm for
$\cM(\prophet_{m+1}, y, \roe)$ that achieves a competitive ratio that
is at least as good.

\subsubsection{Bounding the performance of the oracle model}

After establishing the relationship between our oracle model and the
\textsc{Top-$1$-of-$m$} model, we turn our attention to upper and
lower bounds for the oracle model. First, for the \noniid~setting and
the \roe objective, we present an extremely simple single-threshold
algorithm achieving a competitive ratio that approaches $1$
exponentially in $m$. Even though our algorithm is for the oracle
model, for which weaker guarantees are expected due to
\thmref{second-reduction}, it improves upon the best-known guarantee
for the \textsc{Top-$1$-of-$m$} setting, due to Harb
\cite{h-fbcps-23}. Our algorithm relies on two techniques; sharding
and Poissonization, introduced by \cite{h-fbcps-23} for the analysis
of threshold-based algorithms for prophet inequalities. As an added
benefit, the algorithm's analysis is easy to understand.

Specifically, in \thmref{noniid-cr-asymptotic}, we show that there is
a constant $\xi_m$, such that for the oracle model
$\cM(\oracle_m, \noniid, \roe)$, there exists an algorithm with
competitive ratio at least $1 - e^{-\xi_m}$. As $m \to +\infty$, this
behaves as $1 - e^{- \f{m}{e} + \SmallO\prn{m}}$. The competitive
ratio plot for $m=1,\dots 15$ is shown in \figref{compratio}.

\paragraph{Matching upper bound.}
In addition, we provide a construction for every $m$ that gives a
matching upper bound to the competitive ratio, thus resolving the
problem for the case of general distributions and the~\roe
objective. The construction we have is perhaps of independent interest
in the design of counterexamples for other settings, as it combines
and generalizes standard counterexamples of prophet inequalities.

Specifically, in \thmref{tightness-noniid}, we show that for every
$\delta > 0$, there exists an instance of
$\cM(\oracle_m, \noniid, z)$, where $z = \roe$ or $\pbm$, in which no
single-threshold algorithm can achieve a
$\prn{1 - e^{-\xi_m} + \delta}$-competitive ratio or select the
maximum realization with probability
$\geq \prn{1 - e^{-\xi_m} + \delta}$.

Intuitively, the above follows since an algorithm for the oracle model
performs poorly when, every time it uses an oracle call and gets a
\textsc{YES} answer, the next value it sees that is at least the
queried value is roughly equal, and thus the oracle call was used
without any real gain. The idea behind the worst-case for this setting
is to have what is essentially a Poisson random variable with rate
$\xi_m$, providing the algorithm with several non-zero values, each
roughly the same. By carefully selecting $\xi_m$ in order to equate
the probability of having no non-zero values and the probability of
having more than $m$ non-zero values, we are forcing the algorithm to
use a query for every non-zero realization, thus rendering the oracle
calls as useless as possible.

\subsubsection{The benefit of several oracle calls}
Next, we turn our attention to the~\iid~setting with $m$ oracles calls
and the~\pbm objective. We present a simple, single-threshold
algorithm that selects the maximum realization with probability that
approaches $1$ in a super-exponential fashion. As a warm-up, we first
present the analysis for $m = 1$ before generalizing it to all $m$.

Specifically, in \thmref{prob-asymptotic}, we show that for
$\cM(\oracle_m, \iid, \pbm)$, one can select the maximum realization
with probability at least $1 - \BigO\prn{m^{-\f{m}{5}}}$.

We also present, in \thmref{prob-asymptotic-upper-bound}, an upper
bound on the probability of success that is asymptotically tight, up
to small multiplicative constants in the exponent. Because of
\thmref{equivalence}, both upper and lower bounds on the probability
of success carry over in the \textsc{Top-$1$-of-$m$} settings as well.

See \figref{state:of:art} for a summary of our results for the oracle
model in the different settings.

\begin{figure}
    \centering \includegraphics[width=0.6\textwidth]{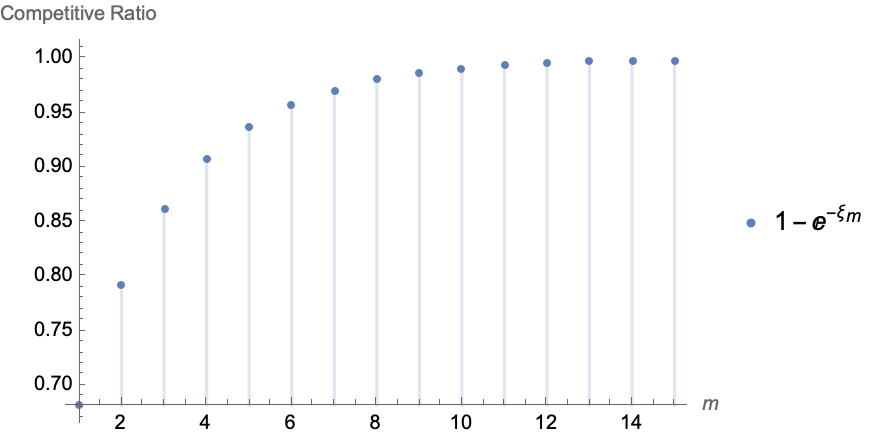}
    \caption{The value of $1-e^{-\xi_m}$ for $m=1,\dots, 15$.}
    \figlab{compratio}
\end{figure}

\begin{figure}
    \begin{tabular}{|c|c|c|c|c|}
      \hline
      \text{Model}
      & \multicolumn{2}{c|}{\text{Lower Bound}}
      & \multicolumn{2}{c|}{\text{Upper Bound}} \\
      \hline
      & \text{\si{Prev}. Best}
      & \text{Current Best}
      & \text{\si{Prev}.}
      & \text{Current Best}
      \\
      \hline
      \hline
      \begin{minipage}{3cm}
          \smallskip%
          \roe\\ General Settings \smallskip%
      \end{minipage}
      & $1 - \BigO\prn{e^{-\f{m}{6}}}$ \: \:
        \text{\cite{efn-pso-18}}
      & $1 - e^{-\f{m}{e} + \SmallO\prn{m}}$
      & -
      & $1 - e^{-\f{m}{e} + \SmallO\prn{m}}$ \text{single-threshold}
      \\
      \hline
      \pbm, \iid Setting
      & $\approx 0.58$ \: \: \cite{gm-rms-66}
      & $
        \begin{array}{@{}c@{}}\approx 0.797 \: \: (m = 1)
          \\
          1 - \BigO\prn{m^{-\f{m}{5}}}\end{array}$
      & -
      & $1 - \BigO\prn{m^{-m}}$
      \\
      \hline
    \end{tabular}
    \caption{State of the art.}
    \figlab{state:of:art}
\end{figure}

\subsection{Additional related work}%
\seclab{related-work} We have already mentioned the related work on
algorithms with predictions, as well as the works of Gilbert and
Mosteller \cite{gm-rms-66}, Esfandiari, Hajiaghayi, Lucier and
Mitzenmacher \cite{ehlm-ps-17} and Nuti \cite{n-spd-22} for the~\pbm
objective. Related work includes the study of order-aware algorithms
by Ezra, Feldman \etal \cite{efgt-winls-23}, algorithms with fairness
guarantees by Correa \etal \cite{ccdn-fbos-21} and algorithms with
a-priori information of some of the values by Correa \etal
\cite{bcgl-gsp-22}. In addition to these, Esfandiari \etal
\cite{ehlm-ps-17} study a related but distinct variant to ours. They
relax the objective to allow the return of one out of the top $k$
realizations, and show exponential upper and lower bounds. Their
model, however, is incomparable to ours.

\paragraph*{Organization.} In \secref{reductions} we relate our model
to \textsc{Top-$1$-of-$m$} model of Assaf and Samuel-Cahn
\cite{as-srpim-00} and prove the reductions. In \secref{non-iid} we
present our tight algorithm for the \noniid setting. \secref{iid}
contains our algorithms and upper bounds for the \iid setting. Due to
space constraints, we present some background on concentration
inequalities that we use for our results
in~\apndref{app:concentration}. Finally, we present \ConfVer{~and
   several missing proofs in~\apndref{app:proofs}}.

\section{Reductions}%
\seclab{reductions}

To motivate our oracle model, we start by establishing an equivalence
between $\cM(\oracle_m, y, \pbm)$ and $\cM(\prophet_{m+1}, y, \pbm)$,
for both the $y =\iid$ and $y =$ \noniid case (see
\thmref{equivalence} below). We also show that, perhaps surprisingly,
this equivalence does not hold for the \roe objective; lower bound
guarantees for $\cM(\oracle_m, y, \roe)$ translate to guarantees for
$\cM(\prophet_{m+1}, y, \roe)$ (\thmref{second-reduction}), but not
the converse. Later, we use this result to improve the best-known
lower bound guarantees on $\cM(\prophet_{m+1}, y, \roe)$.

\subsection{The {\pbmOnly{}} objective}

\begin{lemma}
    \lemlab{prophet-to-oracle}%
    Fix an instance of the prophet problem.
    Let $\cA$ be an algorithm for this instance in
    $\cM(\oracle_m, y, \pbm)$, where $y =\iid$ or \noniid.
    Then, there exists an algorithm $\cB$ for tor this instance in
    $\cM(\prophet_{m+1}, y, \pbm)$, with black-box access to $\cA$,
    such that $\pbm(\cB) \geq \pbm(\cA)$.
\end{lemma}
\begin{proof}
    The idea is for $\cB$ to simulate
    $\cA$'s behavior by selecting each realization that
    $\cA$ decides to query. Initially,
    $\cB$ starts with an empty set
    $S$ of selected values. Whenever
    $\cB$ is presented with a realization $X_i$, it feeds it to
    $\cA$. If $\cA$ decides to select
    $X_i$ or expend a query for
    $X_i$, regardless of the outcome of the query,
    $\cB$ always selects $X_i$ into $S$, otherwise
    $\cB$ decides not to select $X_i$. By induction,
    $S$ contains exactly all the realizations that were queried by
    $\cA$ as well as at most one more realization that might have been
    selected by $\cA$ if it run out of queries. Therefore, $|S| \leq
    m+1$.

    Observe that
    $\cA$ succeeds if and only if it selects the maximum, and it only
    selects a realization $X_i$ if
    $(i)$ it chose to expend a query on $X_i$, or
    $(ii)$ when it observed
    $X_i$ it run out of queries. In both cases, by the description of
    $\cB$, we know that $X_i \in
    S$, and thus the probability that
    $\cB$ succeeds is at least $\pbm(\cA)$.
\end{proof}

\begin{lemma}
    \lemlab{oracle-to-prophet}%
    Fix an input instance of the prophet problem.  Fix an algorithm
    $\cB$ for $\cM(\prophet_{m+1},$ $ y, \pbm)$, where $y =$ \iid or
    \noniid. Then, there exists an algorithm $\cA$ for
    $\cM(\oracle_m, y, \pbm)$, with black-box access to $\cB$, such
    that such that $\pbm(\cA) \geq \pbm(\cB)$.

\end{lemma}
\begin{proof:e}{\lemref{oracle-to-prophet}}{oracle-to-prophet:proof}
    The idea is that $\cA$ can simulate $\cB$'s behavior using the
    oracle queries instead of storing the values like $\cB$
    does. Initially, $\cB$ starts with an empty set $S$ of selected
    values. Whenever $\cA$ is presented with a realization $X_i$, it
    feeds it to $\cB$. If $\cB$ selects $X_i$ into $S$, $\cA$ performs
    an oracle $\oracle$ query whether $X_i > \max_{j = i+1}^n
    X_j$. Consider the first $i$ where this happens. We distinguish
    between the two possible answers: \smallskip%
    \begin{compactenumI}[leftmargin=0.8cm]
        \item If $\oracle$ answers \textsc{YES}, then we know that all
        future realizations are $\leq X_i$. However, we also know that
        since the objective is $\pbm $, any optimal algorithm for
        $\prophet_{m+1}$ will only select a value $X_i$ if it is
        larger than any previously observed value (otherwise it
        ``wastes'' a spot in $S$ for a value that is definitely not
        the maximum). Therefore, if $\cB$ selects $X_i$, we know that
        $X_i > \max_{j < i} X_j$. In this case, both $\cB$ and $\cA$
        succeed in selecting the maximum realization.

        \smallskip%
        \item If $\oracle$ answers \textsc{NO}, then we know that
        there exists a future realization that is $\geq X_i$. In this
        case, the instance for $\cB$ reduces to
        $\cM(\prophet_{m}, y, \pbm)$ on $X_{i+1}, \dots, X_n$, whereas
        the instance for $\cA$ reduces to
        $\cM(\oracle_{m-1}, y, \pbm)$. Since we know that
        $\cM(\prophet_1, y, \pbm) = \cM(\oracle_0, y, \pbm)$ by
        definition, we have that by induction, the probability that
        $\cA$ succeeds is at least $\pbm(\cB)$.
    \end{compactenumI}
\end{proof:e}

Combining the above two lemmas, we get the following result.

\begin{theorem}
    \thmlab{equivalence}
    The $\cM(\oracle_m, y, \pbm)$ model is equivalent to the
    $\cM(\prophet_{m+1}, y, \pbm)$ model, where $y =\iid$ or \noniid.
    In other words, for every prophet inequality instance, the
    probability achieved by the best-possible algorithm in the
    $\cM(\oracle_m, y, \pbm)$ model is the same as the one achieved by
    the best-possible algorithm in the $\cM(\prophet_{m+1}, y, \pbm)$
    model.
\end{theorem}

\subsection{For the \roeOnly objective
   \TPDF{$\oracleOnly_m \leq \prophetOnly_{m+1} $}{oracle diff
      prophet}}

We demonstrate that the $\prophet_{m}$ model strictly surpasses the
$\Oracle_{m}$ for \noniid random variables.

\begin{defn}
    For two integers $i \leq j$, let
    $\IRY{i}{j} = \{ i, i+1, \ldots, j\}$.
\end{defn}
\begin{defn}
    For an instance $\I$ of $\cM(x, \noniid, \roe)$, we denote
    $\roe(x, \I)$ as the competitive ratio of an optimal algorithm for
    $\I$. For example, $\roe(\oracle_m, \I)$ denotes the best
    competitive ratio on instance $\I$ in the oracle model.
\end{defn}

\begin{lemma}
    \corlab{cor:equivalence-counterexample}%
    For $m = 1$, there exists an input instance $\I$, made out of $3$
    \noniid random variables, such that
    \begin{math}
        \roe(\Oracle_1, \I)%
        \leq%
        \frac{3}{4} \roe( \prophet_2, \I).
    \end{math}
\end{lemma}
\begin{proof}
    Consider the input instance $\I$ of three independent random
    variables $X_1, X_2,X_3$ with default value $0$, such that
    \begin{equation*}
        X_1 = 1,%
        \qquad%
        \Prob{X_2 = 1+\eps} =   \frac{1}{2} - \eps, %
        \qquad
        \text{and}
        \qquad
        \Prob{X_3 = \frac{1}{\eps}} =  \eps.
    \end{equation*}
    We have that
    \[
        \Ex{\Bigl.\max\set{X_1, X_2, X_3}}%
        =%
        \frac{1}{\eps} \eps + \prn{1 + \eps} \prn{1 -
           \eps}\prn{\frac{1}{2} - \eps} + 1 \prn{1 -
           \eps}\prn{\frac{1}{2} + \eps}%
        =%
        2-O(\eps).
    \]
    For small $\eps$, an algorithm $\cB$ that is optimal for the
    \defrefY{{$\prophetOnly_2$}}{prop:h:o:m} model in this instance is to
    select $X_1$, ignore $X_2$ and then select $X_3$ if it is
    non-zero. This yields
    \[
        \Ex{\cB}%
        =%
        1 \cdot \prn{1 - \eps} + \frac{1}{\eps} \cdot \eps%
        =%
        2 - \eps.
    \]
    However, the optimal $\cA$ for the oracle model queries $\oracle$
    at $X_1$. With probability $\prn{1 - \eps}\prn{{1}/{2}+\eps}$, it
    stops and select $X_1$, getting a value of $1$. Otherwise, it
    continues, with no oracle calls left. It ignores $X_2$ and select
    $X_3$. Thus,
    \[
        \Ex{\cA} =%
        1 \cdot \prn{\frac{1}{2} + \eps} \prn{1 - \eps} +
        \frac{1}{\eps} \cdot \eps%
        =%
        \frac{3}{2} + \frac{\eps}{2} - \eps^2.
    \]
    The competitive ratios of $\cA$ is
    \begin{math}
        \displaystyle%
        \roe(\oracle_1, \I)%
        =%
        \frac{\frac{3}{2} + \frac{\eps}{2} - \eps^2}%
        {2 - O(\eps)}%
        =%
        \frac{3}{4} + O(\eps)%
        \rightarrow%
        \frac{3}{4}
    \end{math}
    (the limit is for $\eps \rightarrow 0$).  The competitive ratio of
    $\cB$ is
    \[
        \roe(\prophet_2, \I)%
        =%
        \frac{2-\eps}{2+O(\eps)}%
        =%
        1 - O(\eps) \rightarrow 1.
    \]
\end{proof}
The above example, appropriately generalized for $m > 1$ by having
random variables
\begin{align*}
  & X_1 = 1 \quad \text{w.p. } 1, \quad
    X_i = \begin{cases}
      1 + (i-1) \eps & \text{w.p. } \frac{1}{2} - \eps \\
      0 & \text{w.p. } \frac{1}{2} + \eps
    \end{cases}, \quad \text{for } i = 2, \dots, m+1, \text{ and} \\
  & X_{m+2} = \begin{cases}
    \frac{1}{\eps} & \text{w.p. } \eps \\
    0 & \text{w.p. } 1 - \eps
  \end{cases},
\end{align*}
shows that the gap between $\roe(\oracle_m, \I))$ and
$\roe(\prophet_{m+1}, \I))$ is at most $1-1/2^{m+1}$ for general
$m$. The analysis of this example for general $m$ is similar to the
$m = 1$ case. We do not present it here as, even though this example
is very simple, this gap is not the tightest one possible. For a
tighter gap between the competitive ratio of the two models, see the
example in the proof of \thmref{tightness-noniid}.

\begin{lemma}
    \lemlab{prophet-to-oracle-roe}%
    For any input instance $\I$, we have
    $\roe(\prophet_{m+1}, \I) \geq \roe(\oracle_m, \I)$, for $\iid$ or
    $\noniid$ variables.
\end{lemma}
\begin{proof}
    Let $\cA$ be the algorithm in $\cM(\oracle_m, \roe, \I)$ realizing
    the maximum \roe for $\I$. We construct an algorithm
    $\cB \in \cM(\Oracle_m, \roe,\I)$.

    The algorithm $\cB$ simulates $\cA$'s behavior by selecting each
    realization that $\cA$ decides to query. Initially, $\cB$ starts
    with an empty set $S$. Whenever $\cB$ is presented with a
    realization $X_i$, it feeds it to $\cA$. If $\cA$ decides to
    return $X_i$, or performs an oracle query for $X_i$, the algorithm
    $\cB$ adds $X_i$ to $S$.

    Observe that the algorithm $\cA$ stops as soon as an oracle query
    returns \NO. Thus, the simulation $\cB$ of $\cA$, assumes the
    oracle always answers \YES (i.e., a larger value is coming up in
    the future). (i.e., the simulation replaces a call to the oracle
    by a function that always returns \YES), as this enables it
    (potentially) to save more values into the available slots, thus
    increasing its \roe.

    The set $S$ contains exactly all the realizations that were
    queried by $\cA$, as well as at most one additional realization
    returned by $\cA$. Therefore, $|S| \leq m+1$.

    Every possible sequence of realizations $\cA$ queried (or selected
    to return) are in $S$. Therefore, if $V_{\cA}$ is the value
    returned by $\cA$ and $V_{\cB}$ is the value returned by $\cB$, we
    have $V_{\cB} \geq V_{\cA}$, which readily implies that
    $\roe( \cB) \geq \roe(\cA)$.
\end{proof}

\begin{theorem}
    \thmlab{second-reduction} %
    For every $m \geq 1$, and for \emph{all} input instances $\J$ (of
    $\iid$ or $\noniid$ variables), we have
    $ \roe(\oracle_m,\J) \leq \roe(\prophet_{m+1}, \J)$, Furthermore,
    there \emph{exists} an input instance $\I$, made out of $m+2$
    \noniid random variables, such that
    \begin{math}
        \roe(\Oracle_m, \I) \leq (1-1/2^{m+1}) \roe( \prophet_2, \I).
    \end{math}
\end{theorem}

\section{The \noniidOnly settings}
\seclab{non-iid}

By \thmref{second-reduction}, any guarantees we provide for the oracle
model with the \roe objective can be directly translated to guarantees
for the \textsc{Top-$1$-of-$m$} model, improving upon the previous
work on this model \cite{as-srpim-00, ags-rpiwm-02, efn-pso-18,
   h-fbcps-23}. We provide a simple, single-threshold algorithm that
resolves the \roe objective in the oracle model.

\subsection{The exponent sequence}

\begin{definition}
    \deflab{reg-gamma}
    For every $m \geq 1$, let $\xi_m$ denote the unique
    \textit{positive} solution to the following equation:
    \[
        1 - e^{-\xi_m} = \frac{\Gamma(m+1, \xi_m)}{m!},
    \]
    where $\Gamma(m+1,x) = \int_{t=x}^\infty {t^{m} e^{-t} \dif t}$
    denotes the \emph{upper incomplete gamma} function.  The
    \emphi{exponent sequence} is $\xi_1, \xi_2, \ldots$.
\end{definition}
We show below that the optimal competitive ratio of
$\cM(\oracle_m,\noniid, \roe)$ is exactly $1 - e^{-\xi_m}$.  It is
\href{https://en.wikipedia.org/wiki/Incomplete_gamma_function}{known}
that, for $x \geq 0$ and an integer $m+1 > 0$, we have
\begin{equation}
    \Gamma(m+1,x)
    =
    m!\,e^{-x}\sum_{k=0}^{m}{\frac {x^{k}}{k!}}
    \leq
    m! e^{-x} e^x
    \leq
    m!.
    \eqlab{gamma:def}
\end{equation}
As such, the above equation on the value of $\xi_m$, becomes
\begin{equation*}
    1 -e^{-\xi_m}
    =
    e^{-\xi_m}\sum
    _{k=0}^{m}{\frac {\prn{\xi_m}^{k}}{k!}}
    \qquad\iff\qquad%
    \sum
    _{k=m+1}^{\infty}{\frac {\prn{\xi_m}^{k}}{k!}}
    =
    1.
\end{equation*}
This readily implies that the exponent sequence is monotonically
increasing, and $m/e^2 \leq \xi_m \leq m$.

\begin{defn}
    \defnlab{q:function}%
    Let
    \begin{math}
        \QY{k+1}{x} = \frac{\Gamma(k+1, x)}{k!} = e^{-x} \sum_{j =
           0}^k \frac{x^j}{j!}.
    \end{math}
    This implies $\QY{m+1}{\xi_m}=1-e^{-\xi_m}$.
\end{defn}

\begin{lemma}
    \lemlab{silly:d}%
    $\dQY{m+1}{x} = -e^{-x} \frac{x^m}{m!}$.
\end{lemma}
\begin{proof}
    As $(e^{-x})' = -e^{-x}$, we have
    \begin{math}
        \dQY{m+1}{x}%
        = -e^{-x} + \sum_{j = 1}^m\pth{ e^{-x} \frac{x^{j-1}}{(j-1)!}
           -e^{-x} \frac{x^j}{j!}  } =%
        -e^{-x} + e^{-x} -e^{-x} \frac{x^m}{m!}  = -e^{-x}
        \frac{x^m}{m!}.
    \end{math}
\end{proof}

\begin{lemma}
    \lemlab{upper-bound-seq}%
    For all $m\geq 1$, we have
    $ \prn{m!}^{\f{1}{m}} < \xi_m < \prn{(m+1)!}^{\f{1}{m+1}}$.
\end{lemma}
\begin{proof:e}{\lemref{upper-bound-seq}}{upper-bound-seq:proof}
    Guided by \Eqref{gamma:def}, define
    \begin{equation*}
        h(x)
        =
        \frac{\Gamma(m+1, x)}{m!}-1+e^{-x}
        =
        \QY{m+1}{x} -1 + e^{-x}
        =
        e^{-x} \pth{ \sum_{i=0}^m \frac{x^i}{i!} - e^x + 1}
        =%
        e^{-x} \pth{1 -  T(x)},
    \end{equation*}
    where
    \begin{math}
        T(x) = \sum_{i=m+1}^{\infty} \frac{x^i}{i!}.
    \end{math}
    By \lemref{silly:d}, we have
    \begin{math}
        h'(x) = -e^{-x} \frac{x^m}{m!} - e^{-x} < 0.
    \end{math}
    Namely, $h(\cdot)$ is a strictly decreasing function.  Thus,
    $\xi_m$ the positive root of $h(x) = 0$ is unique, as $h(0) = 1$,
    and $\lim_{x\rightarrow \infty} h(x) = -1$.

    Setting $\beta = \prn{(m+1)!}^{1/(m+1)}$, we have
    \begin{math}
        T(\beta) > \frac{\beta^{m+1}}{(m+1)!}  = \frac{(m+1)!}{(m+1)!}
        = 1,
    \end{math}
    which readily implies $h( \beta ) < 0$.  By the AM-GM inequality,
    we have that
    \begin{math}
        \gamma = \sqrt[m]{m!}  < \sum_{i=1}^m i / m = \frac{m+1}{2}.
    \end{math}
    In particular, we have
    \begin{equation*}
        \frac{\gamma^{m+1} }{(m+1)!}
        =%
        \frac{m! \cdot \gamma}{(m+1)!}
        =
        \frac{ \gamma}{m+1}
        <
        \frac{1}{2}.
    \end{equation*}

    As such, we have
    \begin{equation*}
        T(\gamma)
        \leq%
        \sum_{i=m+1}^{\infty} \frac{\gamma^i}{i!}
        \leq
        \sum_{i=m+1}^{\infty}
        \frac{\gamma^{m+1} \gamma^{i-m- 1}}{(m+1)! (m+1)^{i- m -1}}
        <
        \frac{1}{2}
        \sum_{j=0}^{\infty}
        \frac{ \gamma^j}{ (m+1)^j}
        <
        \frac{1}{2}
        \sum_{j=0}^{\infty}
        \frac{ ((m+1)/2)^j}{ (m+1)^j}
        =
        1.
    \end{equation*}
    Thus, $h(\gamma ) > 0$. We conclude that $\gamma < \xi_m < \beta$.
\end{proof:e}

\begin{remark}
    \remlab{decreasing}%
    Setting $\nu(x) = \nu(m, x) = \frac{\Gamma(m+1, x)}{m!}$, and
    arguing as in \lemref{upper-bound-seq}, we have $\nu'(x) < 0$,
    which readily implies that $\nu(x)$ is monotonically decreasing.
\end{remark}

Stirling's formula applied to \lemref{upper-bound-seq} readily implies
the following.

\begin{lemma}
    \lemlab{xi:m}%
    We have
    \begin{math}
        \displaystyle%
        \lim_{m \to \infty} {\frac{\xi_m}{m}} = \frac{1}{e}.
    \end{math}
\end{lemma}

\begin{lemma}
    \lemlab{needed:no:positive}%
    For all $k, m\geq 0$ integers, we have
    \begin{math}
        \displaystyle%
        f(k, m) = \sum_{j=1}^k \frac{\xi_m^j}{j!} - \sum_{j=m+1}^{m+k}
        \frac{\xi_m^j}{j!} \geq 0.
    \end{math}
\end{lemma}
\begin{proof:e}{\lemref{needed:no:positive}}{needed:no:positive:proof}
    By definition $f(0, m)=0$.  We have
    \begin{equation*}
        f(k+1, m)-f(k, m)
        =
        \frac{\xi_m^{k+1}}{(k+1)!} -
        \frac{\xi_m^{m+k+1}}{(m+k+1)!}.
    \end{equation*}
    Thus
    \begin{math}
        f(k+1, m)\geq f(k, m)%
        \iff%
        \f{(m+k+1)!}{(k+1)!} \geq \f{\xi_m^{m+k+1}}{\xi_m^{k+1}} =%
        \xi_m^{m}.
    \end{math}
    for $k > 0$ and $m>0$, we have
    \begin{equation*}
        (m+k+1)!
        =%
        (m+1)! \cdot 1 \cdot  \underbrace{(m+2)}_{> 2}
        \cdot \underbrace{( m+3)}_{> 3} \cdots \underbrace{(m+k+1)
        }_{> k+1}
        >
        (m+1)! (k+1)!.
    \end{equation*}
    Thus, it sufficient to prove that $\xi_m^m < (m+1)! $ $\iff$
    $\xi_m < \sqrt[m]{(m+1)!}$. The later is immediate from
    \lemref{upper-bound-seq}, as
    $\xi_m < \bigl((m+1)!\bigr)^{1/(m+1)} < \sqrt[m]{(m+1)!}$.
\end{proof:e}

\subsection{Background: Sharding, poissonization, and stochastic %
   dominance}

For a sequence of random variables $\XX = X_1, \ldots, X_n$, let
$\cardin{ \alpha \leq \XX \leq \beta} = \cardin{\Set{i }{\alpha \leq
      X_i \leq \beta}}$ denote the number of realizations in this
sequence falling in the interval $[\alpha, \beta]$.

\subsubsection{Sharding}
\seclab{sharding}

For the lower bound, we use poissonization and sharding
\cite{h-fbcps-23}. Given random variables $X_1, \dots, X_n$ with
\cdf{}s $F_1, \dots, F_n$, instead of sampling $X_i$ from $F_i$, we
instead replace it with a sequence of $K$ independent random variables
$\Shard_i = Y_{i,1}, \dots, Y_{i,K}$, such that $\max_j Y_{i,j}$ has
the same distribution as $X_i$. Specifically, the \cdf of $Y_{i,j}$,
for all $j$, is $F^{\f{1}{K}}_i$. Thus, the distribution of
$\max\set{Y_{i,1}, \dots, Y_{i,K}}$ is the same as $X_i$.  This
creates a new sequence of $K n$ samples
$\Seq = \Shard_1 \cdot \Shard_2 \cdots \Shard_n$, where $\cdot$ is the
concatenation operator.  Observe that for any $\alpha \geq 0$ and
integer $t$, we have
\begin{equation*}
    \Prob{\cardin{ \XX \geq  \alpha} > t}
    < \Prob{\cardin{ \Seq \geq \alpha} > t}.
\end{equation*}
This implies, that for threshold algorithms, running on $\Seq$ instead
of $\XX$ can only generate worst results.  We emphasize that this
sharding is done only for analysis purposes.

\subsubsection{Poissonization}

\begin{defn}
    A random variable $X$ has Poisson distribution with rate
    $\lambda$, denoted by $X \sim \Pois( \lambda)$, if
    $\Prob{X = i} = \lambda^k e^{-\lambda}/k!$. Conveniently,
    $\Ex{X} = \Var{X} = \lambda$.
\end{defn}

The purpose of the sharding is to be able to bound quantities of the
form $\Prob{ \cardin{ \beta \leq \Seq \leq \tau}= t}$.  As $K$ grows,
the underlying random variable $\cardin{ \beta \leq \Seq \leq \tau}$
has a binomial distribution that converges to a Poisson distribution.

\begin{observation}
    \obslab{limit}%
    For $c \in (0,1]$, we have, using \LHoptial's rule, that
    $\lim_{x \rightarrow \infty} x(1 - c^{1/x}) = \lim_{x \rightarrow
       \infty} \frac{1 - \exp(\log(c)/x)}{1/x} = \lim_{x \rightarrow
       \infty} \frac{ \log(c)\exp(\log(c)/x)/x^2 }{-1/x^2} = - \log
    c$, where $\log = \log_e$.
\end{observation}

Let $\tau$ be a threshold such that
$ \sum_{i=1}^n \sum_{j=1}^K \Prob{Y_{i,j}\geq \tau} = c$ for some
constant $c$ to be determined shortly. We can rewrite this into the
following.
\begin{equation}
    \sum_{i=1}^n K\prn{1-\Prob{X_i\leq \tau}^{1/K}}  = c.
    \eqlab{sumshards2}
\end{equation}
The limit of \Eqref{sumshards2}, as $K\to +\infty$, is
$ \sum_{i=1}^n -\log \Prob{X_i\leq \tau} = c $, by
\obsref{limit}. Equivalently, for $\ZZ=\max\set{X_1, \ldots , X_n}$,
we have
\begin{equation*}
    e^{-c}%
    =%
    \exp\pth{\Bigl. \smash{\sum_{i=1}^n} \log \Prob{X_i\leq \tau}}
    =%
    \prod_{i=1}^n  \Prob{X_i\leq \tau}
    =
    \Prob{ X_1, \ldots, X_n \leq \tau}
    =%
    \Prob{\ZZ \leq \tau}.
\end{equation*}
In particular, the distribution of the number of indices $j$, such
that $Y_{i,j}\geq \tau$ can be well approximated with a Poisson
distribution. Specifically, let $V_{i,j} =1 $ $\iff$
$Y_{i,j}\geq \tau$, and consider the sum $V_i = \sum_{j=1}^K
V_{i,j}$. The variable $V_i \sim$ $\mathrm{bin}(K, \psi_i)$, where
$\psi_i = 1-\Prob{X_i\leq \tau}^{1/K}$.

Let $\lambda_i = \psi_i K$, and consider the random variable
$U_i \sim \Pois(\lambda_i)$ (i.e., $U_i$ has a Poisson distribution
with rate $\lambda_i$). Intuitively, $V_i$ and $U_i$ have similar
distributions. Formally, \thmrefY{l:e:cam}{Le Cam theorem} implies
that for any set $T \subseteq \{0,1,\ldots, K\}$, we have
\begin{math}
    \abs{\Prob{V_i \in T} - \Prob{U_i \in T}}%
    \leq%
    2K\psi_i^2 =%
    2\lambda_i^2 /K%
    \leq%
    2c^2 /K,
\end{math}
by \Eqref{sumshards2}.  The later quantity goes to zero as $K$
increases.

Thus, we get a variable $U_i$ with a Poisson distribution for each
shard sequence $\Shard_i$, with rate $\lambda_i$, where $U_i$ models
the number of times we encounter in $\Shard_i$ values larger than
$\tau$. Thus, $U_\tau = \sum_i U_i$ models the total number of times
in the splintered sequence $\Seq$ that values encountered are larger
than $\tau$. The variable $U_\tau$ has a Poisson distribution with
rate $\lambda_\tau = \sum_{i=1}^n \lambda_i$.

\subsubsection{The distribution in a range}
\seclab{d:range}

Repeating the same process with a bigger threshold $\beta > \tau$,
would yield a similar Poisson random variable $U_\beta$ with a
\emph{lower} rate $\lambda_\beta$. The quantity
$\Delta =U_\tau - U_\beta$ is the number of values in $\Seq$ in the
range $[\tau, \beta]$. Furthermore, $\Delta$ has a Poisson
distribution with rate $\lambda_\tau - \lambda_\beta$.  Specifically,
$\Prob{ \cardin{ \beta \leq \Seq \leq \tau}= t} = \Prob{\Delta = t }$.

The key to our analysis is that the variables $\Delta$ and $U_\beta$
are independent (in the limit as $K$ increases).

\subsubsection{Stochastic dominance}

A standard observation is that for a non-negative random variable $X$,
we have
\begin{math}
    \Ex{X} = \int_{x=0}^\infty \Prob{X \geq x} \diff x.
\end{math}
Thus, for $\ZZ = \max\set{X_1, \dots, X_n}$, and for an algorithm
$\cA$, if one can guarantee that there is $c \in [0, 1]$, such that
for all $ \nu \geq 0$, $\Prob{\cA \geq \nu} \geq c\Prob{\ZZ\geq \nu}$,
then
\begin{equation*}
    \Ex{\cA}
    =%
    \int_{0}^{\infty} \Pr[\cA \geq x] \diff x
    \geq%
    c \int_{0}^{\infty} \Pr[\ZZ \geq x] \diff x
    \geq  c\Ex{\ZZ}.
\end{equation*}
And hence $c$ is a lower bound on the competitive ratio of $\cA$. This
argument is used in several results on prophet inequalities and is
often referred to as \emph{majorizing} $\cA$ with $\ZZ$.

\subsection{An optimal single-threshold algorithm (lower bound)}

Here, we describe a single-threshold algorithm that achieves the
optimal competitive ratio in the oracle model.
\begin{defn}
    \deflab{single:t}%
    A \emphw{single threshold} algorithm sets a threshold $\tau$, and
    start reading the sequence. Whenever encountering a realization
    $> \tau$, the algorithm stops and consult with the oracle.  The
    oracle query is whether all the values remaining in the suffix of
    the sequence are of value $\leq \tau$.  If the oracle returns
    \YES, the algorithm accepts the current value and stops. Otherwise
    it raises its threshold to $\tau=X_i$ and continues. If the oracle
    runs out of oracle calls, it returns the first value encountered
    after the last oracle call that is bigger than $\tau$ (which
    exists, since all oracle calls returned \NO).
\end{defn}
While technically, the querying threshold of the algorithm might
change during its execution, we call the algorithm a single-threshold
algorithm since it uses a single-threshold to decide whether to query
the oracle or not, and this threshold \emph{does not change with $i$},
unlike for example the optimal DP for the \iid prophet inequality or
the prophet secretary model. Our oracle model is quite different than
most other prophet inequality models in the sense that the algorithm
has some knowledge of the (true) future. Of course, any algorithm that
knows that the maximum of $X_{i+1}, \dots, X_n$ is larger than $X_i$
would be wasting queries if it expended them on some $X_j < X_i$ for
$j > i$, and thus the spirit of it being a single-threshold algorithm
to decide whether to query the oracle or not remains.

\begin{theorem}
    \thmlab{noniid-cr-asymptotic}%
    Let
    $\alpha = 1 - e^{-\xi_m} = 1 - e^{- \f{m}{e} + \SmallO\prn{m}}$,
    see \defref{reg-gamma}.  For any finite sequence $\XX$ of \noniid
    variables, one can compute a value $\tau$, such that the
    single-threshold algorithm (with initial threshold $\tau$) has
    competitive ratio $\geq \alpha$. i.e., the competitive ratio of
    $\cM(\oracle_m, \noniid, \roe)$ is $\geq \alpha$.
\end{theorem}

\begin{proof}
    Let $\XX = X_1, \ldots, X_n$, and $\ZZ = \max_i X_i$. The
    threshold $\tau$ is the $e^{-\xi_m}$ quantile of the maximum, i.e.
    $\Pr[Z \leq \tau] = e^{-\xi_m}$. We use $\cA(\XX)$ to denote the
    result of running the algorithm on $\XX$.

    As suggested in \secref{sharding} (for the analysis), we imagine
    running the algorithm on the splintered sequence $\Seq$. Some
    counterintuitively, imagine first generating $\Seq$, and computing
    $X_i = \max_j Y_{i,j}$, see \secref{sharding}. Thus,
    $\max \Seq = \max \XX$.  For the sequence $\Seq$, let
    $\Seq_{\geq \tau}$ denote the subsequence of elements of $\Seq$
    that their values are above $\tau$. Observe that $\XX_{\geq \tau}$
    is a subsequence of $\Seq_{\geq \tau}$.  Thus, we analyze the
    algorithm performance on $\Seq$.

    Let $\beta \in [0, \tau]$. The probability the algorithm selects a
    value above $\beta$ is equal to the probability it selects any
    value. Thus,
    \begin{equation}
        \Prob{\cA(\XX) \geq \beta}
        =%
        \Prob{\cA \geq \tau}
        =%
        \Prob{Z \geq \tau}
        =
        1-e^{-\xi_m}
        \geq
        \bigl(1 - e^{-\xi_m}\bigr)
        \Prob{Z\geq \beta}.
        \eqlab{firstpart}
    \end{equation}

    For $\beta \in [\tau, +\infty)$, let
    $\Prob{Z\leq \beta} = e^{-q}>e^{-\xi_m}$, implying
    $\Prob{Z\geq \beta} = 1-e^{-q}$.
    By sharding and Poissonization, the number of shards in the range
    $[\tau, \beta]$ (resp. $\geq \beta$) is a Poisson random variable
    $\Delta$ (resp. $U_\beta)$ with rate $\xi_m - q$ (resp. $q$), see
    \secref{d:range}.  Critically, $U_\beta$ and $\Delta$ are
    independent.  Consider the event of there being at most $m$ values
    in the range $[\tau, \beta]$, and there being at least one value
    in $[\beta, +\infty)$.  The value $\cA(\XX) \geq \beta$ in that
    case. Hence, by the independence of $\Delta$ and $U_\beta$, we
    have
    \begin{equation*}
        \frac{ \Prob{\cA(\XX) \geq \beta}}{\Pr[Z\geq \beta]}
        \geq%
        \frac{ \Prob{(U_\beta \geq 1) \cap ( 0\leq \Delta \leq m) }}
        {\Pr[Z\geq \beta]}
        =
        \frac{ \Prob{ U_\beta \geq 1  }} {\Prob{Z\geq \beta} }
        \Prob{ 0\leq \Delta \leq m }
        =
        \Prob{ 0\leq \Delta \leq m }.
    \end{equation*}
    Now, we have
    \begin{equation*}
        \Prob{ 0\leq \Delta \leq m }
        =%
        \sum_{i=0}^m e^{-(\xi_m-q)} \frac{(\xi_m-q)^i}{i!}
        =%
        \frac{\Gamma(m+1, \xi_m-q)}{m!}
        \geq
        \frac{\Gamma(m+1, \xi_m)}{m!}
        =
        1 - e^{-\xi_m}.
    \end{equation*}
    by \Eqref{gamma:def}, \remref{decreasing} and \defref{reg-gamma}.

    The above implies that, for any $\beta \geq 0$, we have
    $\Prob{\cA(\XX) \geq \beta} \geq (1-e^{-\xi_m}) \Prob{Z\geq
       \beta}$, Namely, $\roe( \cA) \geq 1-e^{-\xi_m}$.
\end{proof}

\subsection{A matching upper bound for single-threshold algorithms}

To this end, we present an input sequence for which no algorithm can
do better for the oracle that answers if $X_i > \max_{j=i+1}^n X_j$,
and against an almighty adversary.

\paragraph{Input instance.}
The input instance $\I$ is a sequence made out of $n+2$ random
variables, for $n$ sufficiently large. Each of these random variables
can have only two values -- either zero or some positive
value. Specifically, for $\eps > 0$ sufficiently small (e.g.,
$\eps \ll 1/n^{4}$), let
\begin{equation*}
    X_1 = 1,
    \quad%
    \Prob{X_i = 1 } = \frac{\xi_m}{n},
    \quad\text{for}\quad
    i \in \IRY{2}{n+1},
    \quad
    \text{and}%
    \quad
    \Prob{X_{n+2} = \tfrac{1}{\eps} } = \eps.
\end{equation*}
By \lemref{xi:m}, $\xi_m \approx m/e$, as such, the expected number of
non-zero entries in this sequence is (roughly) $m/e + 1$.

\begin{lemma}
    For $\ZZ = \max_i X_i$, we have $\Ex{\ZZ} =2$ as $\eps \to 0$.
\end{lemma}
\begin{proof}
    Let $\ZZ' = \max_{i\in \IRX{n+1} } X_i$.  Observe that $\ZZ' =
    1$. As such, for $\ZZ = \max(\ZZ', X_{n+2})$, we have
    \begin{math}
        \Ex{\ZZ} =%
        \Ex{\max\nolimits_i X_i}%
        =%
        (1/\eps) \eps + (1 - \eps) \Ex{\ZZ'}
        \xrightarrow[\eps\rightarrow 0]{} 2 .
    \end{math}
\end{proof}

\newcommand{\IX}{\widehat{X}}%

\begin{observation}
    Let $\IX_i$ be an indicator variable for the event that $X_i =
    1$. For sufficiently large $n$, $\nabla = \sum_{i=2}^{n+1} \IX_i$
    has a binomial distribution that can be well approximated by a
    Poisson distribution (see \thmref{l:e:cam}) with rate $\xi_m$.
    That is,
    \begin{math}
        \displaystyle %
        \lim_{n\rightarrow \infty} \Prob{\bigl.\nabla = k} =%
        e^{-\xi_m} \frac{\prn{\xi_m}^k}{k!}.
    \end{math}
\end{observation}

Observe that
$\lim_{n\rightarrow \infty} \Prob{\nabla \leq k } = \sum_{i=0}^k
e^{-\xi_m} \frac{\prn{\xi_m}^i}{i!} = \QY{k+1}{\xi_m}$.  For
simplicity of exposition, we will assume $n\to \infty$ in the
following analysis and thus $\Prob{\nabla \leq k } = \QY{k+1}{\xi_m}$,
see \defnref{q:function}.

\begin{theorem}
    \thmlab{tightness-noniid}
    Consider any choice of $m \geq 1$, and $\delta > 0$, and the above
    input instance $\I$ formed by a sequence of \noniid random
    variables. Then, for any algorithm, against the almighty adversary
    (see \defref{mighty}), we have
    $\cA \in \cM(\oracle_m, \noniid, \roe)$ for $\I$, we have
    $\roe(\cA) \leq 1 - e^{-\xi_m} + \delta$.
\end{theorem}

\begin{proof}
    First, we discuss the strategy that the almighty adversary
    adopts. The adversary first observes all values. Suppose $k$
    nonzero values show up from $X_2, ..., X_n$ at indices
    $U = \{i_1,... , i_k \}$, and all other $n-k$ values from
    $X_2, \dots, X_{n+1}$ at indices
    $B=\{\hat{i}_1, \dots, \hat{i}_{n-k}\}$ are zero. The adversary
    provides the random variables in the order
    $X_{\sigma(1)}, \dots, X_{\sigma(n+2)}$ where $\sigma$ is defined
    as $\sigma(1)=1$, $\sigma(j)=i_j, j=2,\dots, k+1$,
    $\sigma(j)=\hat{i}_j$ and finally $\sigma(n+2)=n+2$. In other
    words, the adversary stacks all the $k$ non zero values from
    $X_2, \dots, X_{n+1}$ starting from index $2$ to index $k+1$.

    Now we consider any algorithm for this setting. We strengthen the
    algorithm by telling it the almighty adversary's strategy; this
    can never reduce its expected reward. Hence, the algorithm knows
    it will see $X_{\sigma(1)}=X_1$, then a stream of $k$ ones (where
    it does not know $k$), then $n-k$ zeros, and finally
    $X_{\sigma(n+2)}=X_{n+2}$. The algorithm has two initial decisions
    to make; either query at $X_1$ and continue (if the answer is \NO)
    to $X_{\sigma(2)}, \dots X_{\sigma(n+1)}$ with $m-1$ oracle calls,
    or it can just proceed to $X_{\sigma(2)}, \dots X_{\sigma(n+1)}$
    with $m$ oracle calls. Thus, the only difference in the two cases
    is that in the former, we have only $m-1$ oracle calls for
    $X_{\sigma(2)}, \dots, X_{\sigma(n+1)}$ but we get an expected
    reward of $1$ if $X_{\sigma(2)}= \dots=X_{\sigma(n+2)}=0$, and in
    the later case, we get $m$ oracle calls for
    $X_{\sigma(2)}, \dots, X_{\sigma(n+1)}$, but we get $0$ reward if
    $X_{\sigma(2)}=\dots = X_{\sigma(n+2)}=0$.

    Let $k$ be the number of non-zeros in $X_{2}, \dots, X_{n+1}$
    (i.e., $X_{\sigma(k+1)}$ is the last $1$). When the algorithm
    starts reading the stream of $1$s from
    $X_{\sigma(2)}, ..., X_{\sigma(n+1)}$, it needs to decide indices
    $S\subseteq \IRY{2}{n+1}, \cardin{S}\leq m$ where it will expend
    the oracle call. Further, it is suboptimal to use the oracle at a
    $0$ value, since regardless, the algorithm will receive a value of
    $0$ in the end if it fails. Consider what happens if the algorithm
    decides to query at index $i\in \IRY{2}{n+1}$ with
    $X_{\sigma(i)}=1$. If $X_{\sigma(i+1)}=\dots X_{\sigma(n+1)}=0$,
    then the algorithm gets on expectation
    $1/\eps \cdot \eps + (1-\eps) \cdot 1 \xrightarrow[\eps\rightarrow
    0]{} 2$ reward on expectation. However, if $X_{\sigma(i+1)}=1$,
    then the oracle will return \NO because
    $1=X_{\sigma(i)} \not > \max(X_{\sigma(i+1)}, \dots,
    X_{\sigma(n+2)})$. On the other hand, if the algorithm does not
    query at index $k+1$ (i.e., $(k+1)\not\in S$), then the algorithm
    gets on expectation
    $\Ex{X_{\sigma(n+2)}} = \Ex{X_{n+2}}=1/\eps \cdot \eps = 1$.

    Hence, the crucial observation is that an algorithm starting at
    $X_{\sigma(2)}$ that uses its query calls at indices
    $S \subseteq \IRY{2}{n+1}$ gets on expectation $2$ if and only if
    $(k+1) \in S$, and $1$ otherwise. Thus, for algorithm $\cA_1$ that
    skips $X_{\sigma(1)}$ and uses its oracles at indices
    $S, \cardin{S}=m$, it satisfies
    \begin{align*}
      \Ex{\cA_1} &= 2\cdot \sum_{i\geq 0, (i+1)\in S} e^{-\xi_m} \frac{\xi_m^{i}}{i!} + 1\cdot \sum_{i\geq 0, (i+1)\notin S} e^{-\xi_m} \frac{\xi_m^i}{i!} \\
                 &=  \sum_{i\geq 0} e^{-\xi_m} \frac{\xi_m^{i}}{i!} +  \sum_{i\geq 0, (i+1)\in S} e^{-\xi_m} \frac{\xi_m^{i}}{i!} \\
                 &= 1 +  \sum_{(i+1)\in S} e^{-\xi_m} \frac{\xi_m^{i}}{i!}
    \end{align*}
    On the other hand, for algorithm $\cA_2$ that uses its oracle at
    $X_{\sigma(1)}$ and uses its remaining oracles at indices
    $S'\in \IRY{2}{n+1}, \cardin{S'}=m-1$, it gets an extra benefit of
    getting a reward with expected value $2$ (as $\eps \to 0$) if
    $X_{\sigma(2)}=\dots = X_{\sigma(n+1)}=0$. Hence, it satisfies
    \begin{align*}
      \Ex{\cA_2} &= \pth{e^{-\xi_m}\cdot  2} +  \pth{2\cdot \sum_{i\geq 0, (i+1)\in S'} e^{-\xi_m} \frac{\xi_m^{i}}{i!}} + \pth{ 1\cdot \sum_{i\geq 1, (i+1)\notin S'} e^{-\xi_m} \frac{\xi_m^i}{i!}} \\
                 &= \pth{\sum_{i\geq 0}e^{-\xi_m}\frac{\xi_m^i}{i!}} + e^{-\xi_m} +  \sum_{(i+1)\in S'} e^{-\xi_m}\frac{\xi_m^i}{i!} \\
                 &= 1 + e^{-\xi_m} + \sum_{(i+1)\in S'} e^{-\xi_m}\frac{\xi_m^i}{i!}.
    \end{align*}

    First, we show that the expression
    $\sum_{(i+1)\in S} e^{-\xi_m}\frac{\xi_m^i}{i!}$ subject to
    $S\subseteq \IRY{2}{n+1}, |S|=m$ is maximized for
    $S^\ast=\IRY{2}{m+1}$. Note that it is easy to verify that for a
    Poisson distribution with rate $\lambda$, its probability mass
    function $e^{-\lambda} \lambda^i/i!$ is increasing for
    $i<\lambda$, and decreasing after $i>\lambda$. Hence, the optimal
    $S^\ast = \IRY{k}{k+m-1}$ for some $k\geq 2$ that ``covers'' the
    rate $\xi_m$ (this is the region with the most mass for a Poisson
    distribution). The optimal choice of $k$ is $k=2$ because
    \begin{align*}
      \sum_{i=1}^{m} e^{-\xi_m}\frac{\xi_m^i}{i!} - \sum_{i=k-1}^{k+m-2} e^{-\xi_m}\frac{\xi_m^i}{i!} &= \sum_{i=1}^{k-2}e^{-\xi_m}\frac{\xi_m^i}{i!} - \sum_{i=m+1}^{m+k-2}e^{-\xi_m}\frac{\xi_m^i}{i!} \geq 0,
    \end{align*}
    where the last inequality holds by
    \lemref{needed:no:positive}. Similarly, $k=2$ is optimal for when
    $|S|=m-1$. Hence, we get the inequalities
    \begin{align*}
      \Ex{\cA_1} &\leq 1 + \sum_{i=1}^m e^{-\xi_m}\frac{\xi_m^i}{i!} = 1 +\QY{m+1}{\xi_m} - e^{-\xi_m},\\
      \Ex{\cA_2} &\leq 1 + e^{-\xi_m} + \sum_{i=1}^{m-1} e^{-\xi_m}\frac{\xi_m^i}{i!} = 1 + \QY{m}{\xi_m}.
    \end{align*}
    Thus, we have
    \begin{align*}
      \max(\Ex{\cA_1}, \Ex{\cA_2})
      &\leq
        1 - e^{-\xi_m}+
        \QY{m}{\xi_m}+
        e^{-\xi_m}\max \{ 1, \frac{\xi_m^m}{m!}  \}
    \end{align*}
    But recall from \lemref{upper-bound-seq} that $\xi_m^m \geq m!$,
    thus
    \begin{align*}
      \max(\Ex{\cA_1}, \Ex{\cA_2})%
      &\leq%
        1 - e^{-\xi_m}+
        \QY{m}{\xi_m}+ e^{-\xi_m}\cdot
        \frac{\xi_m^m}{m!}
      \\&
      =%
      1 - e^{-\xi_m} +  \QY{m+1}{\xi_m}  \\
      &= 2 \pth{ 1 - e^{-\xi_m} }.
    \end{align*}
    Therefore, the competitive ratio of every algorithm is
    \[
        \roe \leq \frac{2\prn{1 - e^{-\xi_m}}}{2} = 1 - e^{-\xi_m}.
    \]
\end{proof}

\begin{remark}
    For a weaker adversary (i.e., offline adversary), one can so very
    slightly better than \thmref{tightness-noniid}.  See
    \apndref{app:final:discussion} for details.
\end{remark}

\section{The \iidOnly settings}
\seclab{iid}

Motivated by the early work of \cite{gm-rms-66} for the
\textsc{Top-$1$-of-$m$} model, in this section we study the \iid
setting and the \pbm objective. As a warm-up, we take a look at the
\iid setting with the \pbm objective and the case of $m = 1$,
providing a simple single-threshold algorithm.

\subsection{A single-threshold algorithm for
   \texorpdfstring{\(m = 1\)}{m = 1}}%

Our single-threshold algorithm $\cA_p$ for
$\cM(\oracle_1, \iid, \pbm)$ selects a threshold $\tau$ equal to the
$p$\th quantile of the given distribution $\cD$, for some
$p \in [0,1]$. In other words, $\tau$ is set such that
$p = \Prob{X_i \geq \tau}$. The first time the algorithm observes a
realization above $\tau$, it queries the oracle to see whether the
realization should be selected or not. If it continues, it simply
accepts the first value encountered above the observed realization on
which it queried $\oracle$.

\begin{lemma}
    There exists $p \in [0, 1]$ such that $\cA_p$ selects the maximum
    realization with probability at least $0.797$ in the
    $\cM(\oracle_1, \iid, \pbm)$ model for large $n$.
\end{lemma}
\begin{proof}
    Let $Y$ be the total number of realizations above $\tau$, and
    $i_1 < i_2 < \dots < i_Y$ be the indices of the random variables
    above $\tau$, i.e. $X_{i_t} > \tau$, for $t = 1, \ldots,
    Y$. Furthermore, let $r_t$ be the rank of $X_{i_t}$ in
    $\cX = \set{X_{i_1}, \dots, X_{i_Y}}$, i.e. the number $k$ such
    that $X_{i_t}$ is the $k$\th largest number in $\cX$, and $Z$ be
    the maximum realization of $X_1, \dots, X_n$.

    $X_{i_1}$ is the first realization we observe above $\tau$. Notice
    that if $r_1 = 1$ or $r_1 = 2$ then the algorithm always selects
    the maximum realization $Z$. In other words, given that $Y = 1$ or
    $Y = 2$, the algorithm selects $Z$ with probability $1$. Consider
    the case $Y > 2$. Again, if $r_1 \leq 2$, the algorithm selects
    $Z$ with probability $1$. Otherwise, if $r_1 > 2$, the algorithm
    returns $Z$ if and only if for all realizations above $\tau$ that
    appear after $X_{i_1}$ and are also larger than $X_{i_1}$, the
    first to encounter is $Z$. In other words, for the algorithm to
    succeed in this case, it must be that among the $r_1 - 1$ values
    of rank smaller than $r_1$, the first one in the arrival order is
    the element of rank $1$. Since the random variables are \iid, the
    probability of this event is exactly $\f{1}{r_1 - 1}$.

    Let $j$ be the first index such that $X_{i_j} > X_{i_1}$, and
    $\alpha(Y) = \ProbCond{\cA \text{ selects } \ZZ }{
       Y}$. Conditioned on $Y \geq 3$, the probability that the
    algorithm selects $Z$ is
    \begin{align*}
      \alpha \prnCond{Y }{ Y \geq 3}
      &=
        \Pr[r_1 = 1] + \Pr[r_1 = 2] + \sum_{t = 3}^Y {\Pr[r_1 = t]
        \ProbCond{r_j = 1 }{ r_1 = t}}
      \\&
      =
      \frac{2}{Y} + \sum_{t=3}^Y \frac{\ProbCond{r_z = 1 }{ r_1 = t}}{Y} \\
      &= \frac{1}{Y} \pth{2 + \sum_{t = 3}^Y \ProbCond{r_z = 1}{ r_1 = t}} \\
      &= \frac{1}{Y} \pth{2 + \sum_{t = 3}^Y \frac{1}{t - 1}} \\
      &= \frac{1}{Y} \pth{1 + \sum_{t = 1}^{Y-1} \frac{1}{t}} \\
      &= \frac{1}{Y} \pth{1 + H_{Y-1}},
    \end{align*}
    where $H_n$ denotes the $n$\th harmonic number. Recall also that
    $\alpha\prnCond{Y }{ Y = 1} = \alpha\prnCond{Y}{ Y = 2} = 1$.

    Next, we estimate $\Pr[Y = i]$, by approximating $Y$ with a
    Poisson distribution via \thmrefY{l:e:cam}{Le Cam's theorem}. Let
    $\delta_i = \abs{\binom{n}{i} p^i (1-p)^{n-i} - e^{-np}
       \frac{\pth{np}^i}{i!}}$. The idea is to set $p$ such that
    $np = q$, where $q\geq 1$ is a fixed constant. We know that
    $\Pr[Y = i] = \binom{n}{i} p^i (1-p)^{n-i}$, and thus, by
    \thmref{l:e:cam}, we have
    \[
        \sum_{i = 0}^\infty {\delta_i} = \sum_{i = 0}^\infty
        {\abs{\Pr[Y = i] - e^{-np} \frac{\pth{np}^i}{i!}}} = \sum_{i =
           0}^\infty {\abs{\Pr[Y = i] - e^{-q} \frac{\pth{q}^i}{i!}}}
        \leq \frac{2 q p}{\max\set{1,q}} \leq 2 p = \frac{2q}{n}.
    \]
    Overall, the probability that $\cA$ selects $Z$ is
    \begin{align}
      \alpha(Y)
      &=
        \sum_{i = 0}^n {\Pr[Y = i]
        \cdot \alpha\prnCond{Y }{ Y = i}} \nonumber \\
      &= \Pr[Y = 1] + \sum_{i = 2}^n {\Pr[Y = i] \cdot
        \alpha\prnCond{Y}{ Y = i}} \nonumber \\
      &\geq n p (1-p)^{(n-1)} + \sum_{i = 2}^n {\pth{e^{-q}
        \frac{q^i}{i!} -
        \delta_i} \cdot \alpha\prnCond{Y }{ Y = i}}, \nonumber
    \end{align}
    where the last inequality follows by the definition of
    $\delta_i$. Thus,
    \begin{align}
      \alpha(Y)
      &=
        q (1 - q/n)^{(n-1)} + \sum_{i = 2}^n {e^{-q} \frac{q^i}{i!}
        \cdot
        \alpha\prnCond{Y }{ Y = i}} - \sum_{i = 2}^n {\delta_i \cdot
        \alpha\prnCond{Y }{ Y = i}} \nonumber
      \\
      &\geq
        q (1 - q/n)^{(n-1)} + \sum_{i = 2}^n {e^{-q} \frac{q^i}{i!}
        \frac{1 +
        H_{i-1}}{i}} - \sum_{i = 2}^n {\delta_i} \nonumber \\
      \eqlab{eq:lecam-forall-n}
      &\geq
        q (1 - q/n)^{(n-1)} + e^{-q}
        \sum_{i = 2}^n {\frac{q^i \pth{1 + H_{i-1}}}{i! \cdot i}}
        - \frac{2q}{n}.
    \end{align}
    It is not too difficult to see after some calculations that, as
    $n \to \infty$, \Eqref{eq:lecam-forall-n} is maximized for
    $q \approx 2.435$, yielding $\alpha(Y) \approx 0.798$.
\end{proof}

It is easy to see that simply setting $q = 2$, which corresponds to
$p = \f{2}{n}$ and $\tau$ being the $\f{2}{n}$\th quantile of $\cD$,
yields $\alpha(Y) > 0.5801$ for all $n \geq 20$. Thus, our simple
single-threshold algorithm, augmented with a single oracle call,
beats, even for small $n$, the optimal algorithm for the \iid prophet
inequality which uses different thresholds per distribution and
achieves a probability of success approximately $0.5801$
\cite{gm-rms-66}.

Since the worst-case probability of $\approx 0.5801$ by
\cite{gm-rms-66} is achieved for $n \to \infty$, one might be
interested in the asymptotic behavior of the probability of our
algorithm, $\alpha(Y)$, for large $n$.

\subsection{A single-threshold algorithm for general
   \texorpdfstring{\(m\)}{m}}

As we saw in the previous section, even for a simple, single-threshold
algorithm, the analysis of the winning probability gets tedious
quickly. In this section, we generalize our single-threshold algorithm
to the case of general $m$, and use the fact that the maximum of a
uniformly random permutation of $n$ values changes $\BigO\pth{\log{n}}$
times with high probability to obtain a guarantee on the winning
probability that is super-exponential with respect to $m$.

As before, our algorithm selects a threshold $\tau$ such that
$p = \Prob{X \geq \tau}$ and every time the algorithm observes a
realization above $\tau$, it uses an oracle query and asks $\oracle$
if the realization should be selected or not. If not, then it updates
the threshold to the new higher value. If the algorithm runs out of
oracle calls, then it selects the first element above the current
threshold $\tau$ that is encounters, if any. In other words, the
algorithm uses the oracle calls greedily for all realizations above
$\tau$.

\begin{theorem}
    \thmlab{prob-asymptotic}%
    For sufficiently large $m, n$, and an instance of
    $\cM(\oracle_m, \iid, \pbm)$, there exists an algorithm that
    selects the maximum realization with probability at least
    $1 - \BigO\pth{m^{-\f{m}{5}}}$.
\end{theorem}

\begin{proof:e}{\thmref{prob-asymptotic}}{prob-asymptotic:proof}
    Let $L = e^{\sqrt{m}}$. The idea is to set $\tau$ so that
    $p = \Prob{X \geq \tau} = \f{L}{n}$. As before, let $Y$ be the
    number of realizations above $\tau$. By \thmref{chernoff}, we
    have
    \[
        \Pr[\abs{Y - L} \geq \delta L] \leq 2 e^{-\delta^2 L/3}.
    \]
    Setting $\delta = 1$ yields that $1 \leq Y \leq 2L$ with
    probability at least
    $1 - 2 e^{-L/3} = 1 - 2 e^{-\f{e^{\sqrt{m}}}{3}} \geq 1 -
    m^{-\f{m}{4}}$ for all $m$.

    Next, let $X'_1, \dots, X'_Y$ be the subsequence of all
    realizations larger than $\tau$, according to their arrival order,
    and let $Z_i = 1$ if $X'_i > \max_{j = 1}^{i - 1} X'_j$, in other
    words if $X'_i$ is larger than all previous realizations, and
    $Z_i = 0$ otherwise. Observe that $\Prob{Z_i = 1} = \f{1}{i}$,
    and that the random variables $Z_1, \dots, Z_n$ are
    independent. Furthermore, let $M = \sum_i Z_i$ be the number of
    times that the maximum realization changes in the sequence
    $X'_1, \dots, X'_Y$. Observe that if $M \leq m+1$, then $m$ oracle
    queries are sufficient for the algorithm to always select the
    maximum realization. Therefore, our goal is to bound the
    probability that this event happens.

    Conditioned on $1 \leq Y \leq 2L$, we have
    \[
        \Ex{M}
        =
        \sum_{i = 1}^{2L} {\frac{1}{i}} \leq \log\pth{2L} + 1
        \leq \sqrt{m} + 2.
    \]
    For $\delta = \f{m+2}{\Ex{M}} - 1$, we have
    \[
        \Pr[M \geq m + 2] = \Prob{M \geq \pth{1+\delta} \Ex{M}}.
    \]
    Notice that for $m \geq 98$, we have $\delta \geq e^2$, and thus,
    by \thmref{chernoff}, we obtain
    \[
        \Pr[M \geq m + 2] \leq e^{ -\f{\Ex{M} \delta\log{\delta}}{2} }
        \leq e^{- \frac{\pth{m - \sqrt{m}}\pth{\log\pth{m - \sqrt{m}}
                 - \f{\log\pth{m+2}}{2}}}{2} } \leq m^{-\f{m}{5}}.
    \]
    If we instead use the tight Chernoff bound in
    \thmref{chernoff}, we can show that
    $\Pr[M \geq m + 2] \leq m^{-\f{m}{4+\eps}}$ for all $m$ and
    $\eps > 0$.

    Putting everything together, for our algorithm to succeed, it
    suffices to have $1 \leq Y \leq 2L$ and $M \leq m + 1$, both of
    which happen together with probability at least
    $1 - \BigO\pth{m^{-\f{m}{5}}}$.
\end{proof:e}

\subsection{An (almost) tight upper bound}

Now that we have presented a simple, single-threshold algorithm for
the $\cM(\oracle_m, \iid, \pbm)$ setting, a reasonable question to ask
is how far it is from being optimal. As we show in this section, the
algorithm is asymptotically almost optimal.

\begin{theorem}
    \thmlab{prob-asymptotic-upper-bound}%
    There exists an instance of $\cM(\oracle_m, \iid, \pbm)$ for which
    no algorithm can select the maximum realization with probability
    greater than $1 - \BigO \pth{m^{-m}}$.
\end{theorem}

\begin{proof:e}{\thmref{prob-asymptotic-upper-bound}}{prob-asymptotic-upper-bound:proof}
    To construct an instance in which no algorithm can achieve a high
    probability, fix $m$ and consider $n$ random variables
    $X_1, \dots, X_n$ drawn \iid from the uniform distribution on
    $[0,1]$, where $n$ is a sufficiently large number. We first divide
    $[0,1]$ into $k = \f{n}{m \log{m}}$ intervals $B_1, \dots, B_k$ of
    length $\f{m \log{m}}{n}$ each, with
    $B_i = \big((i-1) \cdot \f{m \log{m}}{n}, i \cdot \f{m \log{m}}{n}
    \big]$. For each $i = 1, \dots, n$, let $Y_i$ denote the random
    variable that is equal to $1$ if $X_i \in B_k$ and $0$ otherwise,
    where $B_k$ is the last interval. Also, let
    $Y = \sum_{i = 1}^n {Y_i}$. Since the $X_i$'s follow the uniform
    distribution, we have $\Pr[Y_i = 1] = \frac{m \log{m}}{n}$ for all
    $i$, and $\Ex{Y} = m \log{m}$.

    Next, consider an algorithm $\cA$ for $\cM(\oracle_m, \iid, \pbm)$
    on this instance, and assume that $Y \geq 1$, i.e. there exists at
    least one realization that falls in the last interval. Consider
    the moment that $\cA$ observes a realization $X_i \in B_k$ that is
    larger than all previous realizations (including previous
    realizations in $B_k$). There are two cases:
    \begin{compactitem}
        \item If $\cA$ decides not to use a query to $\oracle$ for
        this realization and skip it, there is a chance it fails to
        select the highest realization. This definitely happens if no
        other realization in the future is in $B_k$, which occurs with
        probability
        \[
            \pth{1-X_i}^{n-i} \geq \pth{1 - \frac{m \log{m}}{n}}^{n -
               i} \geq \pth{1 - \frac{m \log{m}}{n}}^n \geq e^{-m
               \log{m}-1} = \BigOmega\pth{m^{-\f{m}{1-\eps}}}
        \]
        for sufficiently large $n$, for any $\eps > 0$.

        \smallskip%
        \item If $\cA$ decides to expend a query to $\oracle$ for this
        realization, there is a chance it fails to select the highest
        realization by running out of queries, deciding to select the
        next realization in $B_k$ that is higher than all previous
        ones, and missing out on a higher realization in the
        future. For this to happen, it must be that $Y \geq m+2$. Let
        $\delta = 1 - \f{\left(1+\f{1}{m}\right)}{\log{m}}$. By
        \thmref{chernoff}, this happens with probability
        \begin{align*}
          \Pr[Y > m + 1]
          &=%
            1 - \Pr[Y \leq m+1]
          =%
            1 - \Pr[Y \leq (1-\delta) \Ex{Y}]%
          \geq%
            1 - e^{-\frac{m \log{m} \pth{\log{m} - 1
            - \f{1}{m}}^2}{2 {\log{m}}^2}} \\
          &\geq%
            1 - m^{-\f{m}{4}}.
        \end{align*}
        Given that $Y \geq m+2$, the probability that the first $m+2$
        realizations arrive in increasing order is
        $\f{1}{(m+2)!}$. Therefore, $\cA$ misses out on the maximum
        realization in this case with probability at least (for
        $m\geq 6)$
        \[
            \frac{1 - m^{-\f{m}{4}}}{(m+2)!} \geq m^{-m}.
        \]
    \end{compactitem}
    Therefore, $\cA$ must miss the maximum realization with
    probability at least $\BigOmega \pth{m^{-m}}$.
\end{proof:e}

\BibTexMode{%
   \ConfVer{%
      \bibliographystyle{plainurl}%
   } \NotConfMode{%
      \bibliographystyle{alpha}%
   } \bibliography{prophet_oracle} }%
\BibLatexMode{\printbibliography}

\appendix

\section{Some tools from probability}
\apndlab{app:concentration}

\begin{theorem}[Chernoff's inequality \cite{dp-cmara-09}]
    \thmlab{chernoff} Let $Y_1, \dots, Y_n$ be independent indicator
    random variables with $p_i = \Prob{Y_i = 1}$ and $Y = \sum_i
    Y_i$. Let $\mu = \Ex{Y} = \sum_i p_i$. Then,
    \begin{compactenumI}
        \item For $\delta \geq 0$:
        \begin{math}
            \Prob{Y \geq (1 + \delta) \mu} \leq
            \pth{\frac{e^{\delta}}{(1 + \delta)^{(1+\delta)}}}^{\mu}.
        \end{math}

        \smallskip%
        \item For $\delta \geq 0$:
        \begin{math}
            \Prob{Y \leq (1 - \delta) \mu} \leq
            \pth{{\frac{e^{-\delta}}{(1-\delta)^{(1-\delta)}}}}^{\mu}.
        \end{math}

        \smallskip%
        \item For $\delta\in (0,1]$,
        \begin{math}
            \Prob{Y \geq (1 + \delta) \mu} \leq e^{-\mu
               \f{\delta^2}{3}}.
        \end{math}

        \smallskip%
        \item For $\delta \in (0,1]$
        \begin{math}
            \Prob{Y \leq (1 - \delta) \mu} \leq e^{-\mu
               \f{\delta^2}{2}}.
        \end{math}

        \smallskip%
        \item For $\delta > e^2$,
        \begin{math}
            \Prob{Y \geq (1 + \delta) \mu} < e^{-\frac{\mu \delta
                  \log{\delta}}{2}}.
        \end{math}
    \end{compactenumI}
\end{theorem}

The following is known as Le Cam's theorem, see
\cite{c-atpbd-60,d-h-ptcm-12}.
\begin{theorem}[Le Cam's theorem]
    \thmlab{l:e:cam}%
    Let $X_1, \ldots, X_n$ be independent Bernoulli random variables,
    with $p_i=\Prob{X_i=1}$, for $i \in \IRX{n}$. Let $S = \sum_i X_i$
    and $\lambda = \sum_i p_i$.  Then $S$ has a Poisson binomial
    distribution with expectation ${\lambda}$. Furthermore, let
    \begin{math}
        Y \sim \Pois \lambda.
    \end{math}
    Then we have
    \begin{equation*}
        \sum_{i=0}^n \cardin{\Prob{S =i} - \Prob{Y=i}}
        =%
        \sum_{i=0}^n \cardin{\Prob{S =i} - e^{-\lambda}
           \frac{\lambda^i}{i!} }
        \leq
        2 \sum_{i=1}^n p_i^2.
    \end{equation*}
\end{theorem}

\section{Discussion on variants of the oracle and adversary choice. }
\apndlab{app:final:discussion}

In this discussion, we will weaken the
adversary from an almighty adversary to an offline adversary. This can
only increase the upper bound. We will also strengthen the oracle to
answer $\geq$ instead of $>$ oracles. Note that the sharding analysis
still holds, and the lower bound of $1-e^{-\xi_m}$ still holds. We
show that
\begin{enumerate}
    \item Amongst single-threshold algorithms, our algorithm is still optimal. That is, no single threshold algorithm can get a competitive ratio $\geq 1-e^{-\xi_m}$.
    \item However, if we allow the algorithm to use thresholds that depend on $i$ in the sequence $X_1, \dots, X_n$, then our single-threshold ceases to be optimal. However, we show that the discrepancy from an optimal algorithm in this scenario, to an optimal single-threshold algorithm is very small\footnote{We observed that the maximum difference between the single-threshold
optimal competitive ratio and the optimal multiple-threshold
competitive ratio is at most $\leq 0.0018$ for all $m$. In fact for $m=1$, the discrepancy is $\leq 0.000000973$! See
\figref{discrepancy}. This very small discrepancy is separately
interesting to address, especially that the discrepancy between
optimal multiple-threshold algorithms and single-threshold algorithms
for other prophet inequalities is usually quite large. For example,
the \textsc{Top-$1$-of-$2$} model has a discrepancy of almost $0.1$ between single threshold and multiple-threshold algorithms.}
\end{enumerate}

\paragraph{Input instance.}
The input is a sequence made out of $n+2$ random variables, for $n$
sufficiently large. Each of these random variables can have only two
values -- either zero or some positive value. Specifically, for
$\eps > 0$ sufficiently small (e.g., $\eps \ll 1/n^{4}$), let
\begin{equation*}
    X_1 = 1,
    \quad%
    \Prob{X_i = 1 + \eps  (i-1)} = \frac{\xi_m}{n},
    \quad\text{for}\quad
    i \in \IRY{2}{n+1},
    \quad
    \text{and}%
    \quad
    \Prob{X_{n+2} = \tfrac{1}{\eps}\Psi } = \eps,
\end{equation*}
where
\begin{math}
    \displaystyle%
    \Psi ={m!}/{\xi_m^m}.
\end{math}
By \lemref{xi:m}, $\xi_m \approx m/e$, as such, the expected number of
non-zero entries in this sequence is (roughly) $m/e + 1$.

\begin{observation}
    We have that $\Psi ={m!}/{\xi_m^m} < 1$ by
    \lemref{upper-bound-seq}.
\end{observation}

\begin{lemma}
    For $\ZZ = \max_i X_i$, we have $\Ex{\ZZ} =1 + \Psi$, where
    $\Psi ={m!}/{\xi_m^m}$.
\end{lemma}
\begin{proof}
    Let $\ZZ' = \max_{i\in \IRX{n+1} } X_i$.  Observe that
    $1 = X_1 \leq \ZZ'\leq 1+\eps n$.  Thus,
    $\lim_{\eps \rightarrow 0} \ZZ' = 1$.  As such, for
    $\ZZ = \max(\ZZ', X_{n+2})$, we have
    \begin{math}
        \Ex{\ZZ} =%
        \Ex{\max\nolimits_i X_i}%
        =%
        (\Psi/\eps) \eps + (1 - \eps) \Ex{\ZZ'}
        \xrightarrow[\eps\rightarrow 0]{} \Psi + 1 .
    \end{math}
\end{proof}

\begin{observation}
    Let $\IX_i$ be an indicator variable for the event that
    $X_i \geq 1$. For sufficiently large $n$,
    $\nabla = \sum_{i=2}^{n+1} \IX_i$ has a binomial distribution that
    can be well approximated by a Poisson distribution (see
    \thmref{l:e:cam}) with rate $\xi_m$.  That is,
    \begin{math}
        \displaystyle %
        \lim_{n\rightarrow \infty} \Prob{\bigl.\nabla = k} =%
        e^{-\xi_m} \frac{\prn{\xi_m}^k}{k!}.
    \end{math}
\end{observation}

Observe that
$\lim_{n\rightarrow \infty} \Prob{\nabla \leq k } = \sum_{i=0}^k
e^{-\xi_m} \frac{\prn{\xi_m}^i}{i!} = \QY{k+1}{\xi_m}$.  For
simplicity of exposition, we are going to pretend that
$\Prob{\nabla \leq k } = \QY{k+1}{\xi_m}$.

\begin{theorem}
    \thmlab{tightness-noniid:apdx}
    Consider any choice of $m \geq 1$, and $\delta > 0$, and the above
    input instance $\I$ formed by a sequence of \noniid random
    variables. Then, for all \textbf{single-threshold algorithm}
    $\cA \in \cM(\oracle_m)$ for $\I$, we have
    $\roe(\cA) \leq 1 - e^{-\xi_m} + \delta$ and
    $\pbm(\cA) \leq 1 - e^{-\xi_m} + \delta$.
\end{theorem}

\begin{proof}
    Consider all the distinct values that might appear in $\I$ --
    there are $n+2$ such values. Thus, there are only $n+2$
    single-threshold algorithms we need to consider (corresponding to
    each of these values). If the threshold is $\tau = \Psi/\eps$,
    then the algorithm gets on expectation $\Psi < 1$, which is
    clearly suboptimal compared to accepting $X_1$ immediately.

    Next, consider the threshold $1$. If there are at most $m-1$
    non-zero values in $X_2, \ldots, X_{n+1}$, then the algorithm
    continues to $X_{n+2}$ and the algorithm gets at most
    $\max(X_{n+2}, 1+n\eps)$. If there are at least $m$ non-zero
    values in $X_2, \ldots, X_{n+1}$, then the algorithm gets at most
    $1+n\epsilon$. Hence, for
    \begin{equation}
        \rho%
        =%
        \Psi \frac{1}{\eps} \cdot \eps +
        (1-\eps)(1+n\eps)
        \xrightarrow{\eps \to 0}
        1 + \Psi,
        \eqlab{rho}
    \end{equation}
    we have
    \begin{align*}
      \Ex{\cA}
      &\leq%
        \Prob{\nabla  \leq  m-1 }
        \pth{
        \Psi \frac{1}{\eps} \cdot \eps +
        (1-\eps)(1+n\eps)
        }
        +
        \Prob{\nabla  \geq  m }
        \pth{1+n\eps}
        \smash{\underbrace{\leq}_{\eps \rightarrow 0 }}
        \QY{m}{\xi_m}
        \Psi  + 1 .
    \end{align*}
    Since $\QY{m+1}{\xi_m}=1-e^{-\xi_m}$, we have
    \begin{equation*}
        \QY{m}{\xi_m} \Psi
        =%
        \QY{m}{\xi_m} \frac{m!}{\xi_m^m}
        =
        \Bigl( \QY{m+1}{\xi_m}-e^{-\xi_m}\frac{\xi^m}{m!} \Bigr)
        \cdot
        \frac{m!}{\xi_m^m}
        =%
        (1-e^{-\xi_m})\Psi - e^{-\xi_m}.
    \end{equation*}

    Thus, we have (as $\eps \to 0$) that
    \begin{equation*}
        \Ex{\cA }
        \leq%
        {1-e^{-\xi_m} +  (1-e^{-\xi_m})\Psi}
        =%
        (1-e^{-\xi_m})\pth{1+\Psi}.
    \end{equation*}

    Next, consider a threshold of $1+ (i-1)\eps$. Here, the algorithm
    is ``activated'' at $X_i$, for $ i\in \IRY{2}{n+1}$. Let
    \begin{equation*}
        \beta = (n-i+2)\frac{\xi_m}{n}
    \end{equation*}
    denote the Poisson rate for $X_i, \ldots, X_{n+1}$, and let
    $\nabla_i = \sum_{j=i}^{n+1} X_j$. Clearly $\beta\leq \xi_m$. If
    $X_i, \ldots, X_{n+1}$ are all zeros, then the algorithm gets the
    expectation of $X_{n+2}$. If there are at least one, and at most
    $m$ non-zero values in $X_i, \ldots, X_{n+1}$, then the algorithm
    continues to $X_{n+2}$ and the algorithm gets at most
    $\Ex{\max(X_{n+2}, 1+n\eps)} \leq \rho = 1+\Psi$, see \Eqref{rho}.
    If there are at least $m+1$ non-zero values in
    $X_i, \ldots, X_{n+1}$, then the algorithm gets at most
    $1+n\epsilon$. Hence, on expectation, the algorithm gets
    \begin{align*}
      \Ex{\cA}
      &\leq%
        \Prob{\nabla_i = 0} \Ex{X_{n+2}}
        +
        \Prob{\bigl.\nabla_i \in \IRY{1}{m}}\rho
        +
        \Prob{\bigl.\nabla_i > m} (1+n \eps)
      \\&%
      \leq
      \QY{1}{\beta} \Psi +
      \pth{\QY{m+1}{\beta} - \QY{1}{\beta}} ( 1+\Psi)
      +
      (1-\QY{m+1}{\beta}) \cdot \pth{1+n\eps}
      \\
      &\xrightarrow[\eps\rightarrow 0]{}
        \QY{1}{\beta} \Psi +
        \QY{m+1}{\beta} +
        \QY{m+1}{\beta} \Psi
        - \QY{1}{\beta}   - \QY{1}{\beta}\Psi
        +
        1-\QY{m+1}{\beta}
      \\&%
      =%
      1 + \QY{m+1}{\beta} \Psi - \QY{1}{\beta}
      =%
      1 + \QY{m+1}{\beta} \Psi - e^{-\beta}.
    \end{align*}
    Let
    \begin{math}
        f(x) = \QY{m+1}{x} \Psi - e^{-x}.
    \end{math}
    By \lemref{silly:d}, we have
    \begin{math}
        f'(x) = -e^{-x} \frac{x^m}{m!} \Psi + e^{-x} =%
        -e^{-x} \frac{x^m}{\xi_m^m} + e^{-x}.
    \end{math}
    As such, $f'(x) >0$ for $x \in [0,\xi_m)$. Namely, $f$ is
    increasing in this range, and this range contains the value
    $\beta$.  Since $\QY{m+1}{\xi_m}=1-e^{-\xi_m}$, we have
    \begin{equation*}
        \Ex{\cA}
        \leq
        1 + f(\beta)
        \leq
        1 + f(\xi_m)
        \leq
        1 + (1-e^{-\xi_m}) \Psi - e^{-\xi_m}
        =
        (1- e^{-\xi_m})(1 + \Psi).
       \end{equation*}

    This implies, that for all cases, the competitive ratio is
    \begin{equation*}
        \frac{ \Ex{\cA} }{\Ex{Z}}
        \leq%
        \frac{1-e^{-\xi_m} +  (1-e^{-\xi_m})\Psi}{1+\Psi}
        =%
        \frac{\pth{1-e^{-\xi_m}}\pth{1+\Psi}}{1+\Psi}
        =%
        1-e^{-\xi_m}.
    \end{equation*}
\end{proof}

\subsection{On the optimal multiple-threshold algorithms}

Intuitively, a better strategy than single-threshold, is using
different thresholds (potentially based on the values seen so far).

\paragraph{New Input Instance}
The input is a sequence made out of $n+2$ random variables, for $n$
sufficiently large. Each of these random variables can have only two
values -- either zero or some positive value. Specifically, for
$\eps > 0$ sufficiently small (e.g., $\eps \ll 1/n^{4}$), let
\begin{equation*}
    X_1 = 1,
    \quad%
    \Prob{X_i = 1 + \eps  (i-1)} = \frac{c_{1,m}}{n},
    \quad\text{for}\quad
    i \in \IRY{2}{n+1},
    \quad
    \text{and}%
    \quad
    \Prob{X_{n+2} = \tfrac{c_{2,m}}{\eps} } = \eps,
\end{equation*}
\paragraph{Optimal algorithm.}

For a fixed instance of the prophet inequality problem, one can
usually argue about the optimal algorithm for the instance using
reverse dynamic programming. The argument is standard, but we include
it here for the sake of completeness. Let
$b_i=1+\eps(i-1)$, for $i\in \IRY{1}{n+1}$, and
$b_{n+2}=c_{2,m}/\eps$.  Note that $b_1\leq \cdots \leq b_{n+2}$.
Similarly, let $p_1=1$, $p_i = c_{1,m}/n$ for $i\in \IRY{2}{n+1}$, and
finally $p_{n+2} = \eps$. Now, the input is the sequence of random
variables $X_1,\ldots, X_{n+2}$, with
\begin{equation*}
    \Prob{X_i = b_i } = p_i,
    \quad\text{for}\quad
    i=1,\ldots, n+2.
\end{equation*}
Let $Z_i = \max(X_i, \ldots., X_n)$. Let $E_t(k)$ be the expected
value of an \textit{optimal algorithm} running on $X_k,\ldots,X_{n}$
having access to $t$ oracle calls. Let $E_t^{\uparrow}(k)$ be the
expected value of an \textit{optimal algorithm} running on
$X_k,\ldots, X_{n}$ with $t$ oracle calls, given that $Z_k>0$. We have
the (mutual) recurrence
\begin{align*}
  E_t(k)
  =\,\,
  & \Pr[X_k=0]E_t(k+1) \\&
  +\,%
  \Pr[X_k>0]
  \begin{cases}
    \max
    \begin{cases}
      E_t(k+1)\\
      \Pr[Z_{k+1}=0]b_k + \Pr[Z_{k+1}>0]E_{t-1}^{\uparrow}(k+1)
    \end{cases}
    & t>0 \\
    \max \left( E_t(k+1), b_k \right)
    & t=0.
  \end{cases}
\end{align*}
Let
\begin{equation*}
    \alpha_k
    =%
    \ProbCond{X_k>0}{Z_k>0}
    =%
    \frac{\Prob{(X_k>0) \cap (Z_k >0) }}
    {\Prob{Z_k > 0 }}
    =%
    \frac{\Prob{X_k>0  }}
    {\Prob{Z_k > 0 }}.
\end{equation*}
Consider
\begin{equation*}
    \beta_{k+1}
    =
    \ProbCond{Z_{k+1}>0}{(X_k > 0)  \cap (Z_k > 0 )}
    =%
    \ProbCond{Z_{k+1}>0}{X_k > 0}
    =%
    \Prob{Z_{k+1}>0}.
\end{equation*}
We now have
\[
    E_t^{\uparrow}(k)
    =%
    (1-\alpha_k) E_t^\uparrow(k+1)
    +
    \begin{cases}
      \alpha_k \max
      \begin{cases}
        E_t(k+1),\\
        (1-\beta_{k+1})b_k
        +
        \beta_{k+1}E_{t-1}^{\uparrow}(k+1)
      \end{cases}
      & t>0
      \\
      \alpha_k \max \left( E_t(k+1), b_k\right)
      &
        t=0
    \end{cases}
\]
It is straightforward to argue by reverse induction on $k$ that
$E_t(k)$ is the best any algorithm can get on expectation from
$X_k, \ldots, X_{n+2}$ using $t$ oracle queries, as the recurrence
includes all possible outcomes of these queries. Also note that the
recurrence can easily be evaluated in $O(nm)$ time.

\paragraph{How far off is the optimal single-threshold algorithm} It
is natural to ask how far our optimal single-threshold algorithm is
from the optimal multiple-threshold algorithm above. We set $n=100000$
and $\eps=10^{-18}$, and performed experiments on $m=1, \ldots, 11$.
We observed that the maximum difference between the single-threshold
optimal competitive ratio and the optimal multiple-threshold
competitive ratio is at most $\leq 0.0018$. See
\figref{discrepancy}. This very small discrepancy is separately
interesting to address, especially that the discrepancy between
optimal multiple-threshold algorithms and single-threshold algorithms
for other prophet inequalities is usually quite large. For example,
the \textsc{Top-$1$-of-$2$} model has a discrepancy of almost $0.1$ between single threshold and multiple-threshold algorithms.

\begin{table}
    \centering
    \begin{tabular}{|c|c|c|c|c|c|}
    \hline
    \hline
      $m$ & $c_{1,m}$ & $c_{2,m}$  & OPT Competitive Ratio
      & $1-e^{-\xi_m}$ & \textbf{Difference}\\
        \hline
        \hline
1  &  1.146  &  0.872  &  0.682  &  0.682  &  0.000000973 \\
\hline
2  &  1.685  &  0.779  &  0.792  &  0.792  &  0.000689 \\
\hline
3  &  2.054  &  0.808  &  0.863  &  0.861  &  0.00178 \\
\hline
4  &  3.250  &  0.682 &  0.909  &  0.907  &  0.00170 \\
\hline
5  &  3.696  &  0.651  &  0.939  &  0.937  &  0.00170 \\
\hline
6  &  3.826  &  0.628  &  0.959  &  0.958  &  0.00154 \\
\hline
7  &  4.330  &  0.612  &  0.973  &  0.971  &  0.00131 \\
\hline
8  &  4.195  &  0.682  &  0.982  &  0.980  &  0.00113 \\
\hline
9  &  5.234  &  0.580  &  0.988  &  0.987  &  0.000846 \\
\hline
10  &  5.854  &  0.571  &  0.992  &  0.991  &  0.000656 \\
\hline
11  &  6.131  &  0.563  &  0.994  &  0.994  &  0.000500 \\
\hline
    \end{tabular}
    \caption{Maximum discrepancy between single-threshold algorithm and multiple-threshold algorithm for $m=1$ to $m=11$. Rounded to $3$ significant digits. }
    \figlab{discrepancy}
\end{table}

\InsertAppendixOfProofs

\end{document}